%% file: minimax-theorems-joint-source-channel-coding-AVC-arxiv.tex
\documentclass[journal, onecolumn, 12pt, draftclsnofoot]{IEEEtran}

\usepackage{amsmath,amsfonts,amssymb,float,enumerate,bm,graphicx,mathtools}

\DeclarePairedDelimiter{\ceil}{\lceil}{\rceil}
\DeclarePairedDelimiter{\floor}{\lfloor}{\rfloor}

\allowdisplaybreaks

\usepackage{tikz}
\usetikzlibrary{arrows}
\usepackage{verbatim}
\tikzstyle{int}=[draw, fill=white!20, minimum size=2em]
\tikzstyle{init} = [pin edge={to-,thin,black}]

\usepackage[left=1.27cm,right=1.27cm,bottom=2.9cm,top=1.9cm,a4paper]{geometry}

\restylefloat{table}
\newcommand{\QXS}{Q_{X|S}}
\newcommand{\QSY}{Q_{\widehat{S}|Y}}

\newcommand{\wi}[1]{\widehat{#1}}
\newcommand{\mbb}[1]{\mathbb{#1}}
\newcommand{\mcal}[1]{\mathcal{#1}}

\def\SC{{\rm SC}}
\def\LP{{\rm LP}}
\def\OPT{{\rm OPT}}

\input{macros-ieee}

\ifCLASSINFOpdf
  \else
 
\fi

\title{Minimax Theorems for Finite Blocklength Lossy Joint Source-Channel Coding over an AVC}
\author{\hspace{0.5cm} Anuj S. Vora, \and Ankur A. Kulkarni\thanks {The authors are with the Systems and Control Engineering group at the Indian Institute of Technology Bombay, Mumbai, 400076. They can be reached at \texttt{anujvora@iitb.ac.in},  \texttt{kulkarni.ankur@iitb.ac.in}. An earlier version of this work was presented at the National Conference on Communications in Bangalore~\cite{vora2019minimax}.}}

\begin{document}
\maketitle
\vspace{-1cm}
\begin{abstract}
Motivated by applications in the security of cyber-physical systems, we pose the  finite blocklength communication problem in the presence of a jammer as a  zero-sum game between the encoder-decoder team and the jammer, by allowing the communicating team as well as the jammer only locally randomized strategies. 
The communicating team's problem is non-convex under locally randomized codes, and hence, in general, a minimax theorem need not hold for this game. However, we show that \textit{approximate minimax} theorems hold in the sense that the minimax and maximin values of the game approach each other asymptotically. In particular, for rates strictly below a critical threshold, \textit{both} the minimax and maximin values approach zero, and for rates strictly above it, they both approach unity. We then show a \textit{second order} minimax theorem, \ie, for rates exactly approaching the threshold with along a specific scaling,  the minimax and maximin values approach the same constant value, that is neither zero nor one. Critical to these results is our derivation of finite blocklength bounds on the minimax and maximin values of the game and our derivation of second order dispersion-based bounds.
\end{abstract}

\section{Introduction}

Cyber-physical systems consist of physical entities that are remotely controlled via communication channels. Examples of such systems are smart city infrastructure, modern automobiles, power grids and nuclear power plants where a background cyber layer \cite{humayed2017cyber} is deployed to carry time-critical information. 
Efficient and secure transmission of information over the cyber layer is fundamental to the functioning of such  systems. To guarantee reliable communication in the block coding paradigm, it generally requires an increasingly large number of uses of the channel, which induces  delays in the transmission. Since the systems above are sensitive to delays, it is of utmost importance to study these systems in the finite blocklength regime. On another count, the presence of the cyber layer in these systems makes them vulnerable to adversarial attacks which can have catastrophic consequences \cite{slay2007lessons, langner2011stuxnet}. 
 Motivated by these two factors -- limited delay and the possibility of jamming --  in this paper, we study a point-to-point finite blocklength lossy communication problem in presence of a jammer. 

We consider a setting where a sender and a receiver are communicating over a channel whose state is controlled by an active jammer. The jammer can vary the state \textit{arbitrarily}, \ie, at every instance of transmission, with the goal of disrupting communication. The standard formulation of this problem considers this setting only from the vantage point of the communicating team, by seeking coding strategies that are robust against \textit{any} action of the jammer. However, since both the communicating team and the jammer have strategic capabilities, it is logical that one adopts a neutral point of view where one allows both the communicating team and the jammer to act strategically. To this end, we formulate the above problem as a zero-sum game where the sender-receiver team aims to minimize the loss in communication by designing suitable codes, while the jammer tries to maximize it by choosing channel states. 

Critical to our contribution is the definition of allowable strategies for the communicating team and the jammer; in particular the notion of randomization we permit. We allow the strategy for the communicating team to be a \textit{stochastic code}, that is, a pair of encoder and decoder, that are allowed only \textit{local} randomization. 
Generally, for a communication problem, it is of interest to determine codes which are deterministic. In the game theoretic parlance, this corresponds to \textit{pure} strategies for the communicating team. However, for reasons discussed later, finding a deterministic code for this problem appears to require analysis which is beyond the scope of this paper. Consequently, we look for randomized strategies. 
Game-theory provides us with two possible kinds of randomizations of pure strategies (see \eg,~\cite{maschler2013game}). A \textit{mixed} strategy is a random choice of a pure strategy; in the communication theoretic parlance this corresponds to a randomly chosen pair of encoder and decoder, or \textit{random code}. However, implementing such a code requires \textit{shared} randomness between the encoder and the decoder, which is may not be feasible in cyber physical systems where the encoder and decoder are decentralized with the channel as a sole medium of communication. Another notion of randomization in game theory is \textit{behavioral} strategies -- under such a strategy, an action is chosen at random given the information at the time. In the information theory literature these are referred to as \textit{stochastic codes}. Under such a strategy, the encoder and decoder randomize \textit{locally} using their own private source of randomness. We adopt this notion of randomization in this paper for the communicating team.  For the jammer, we allow it to randomize over the choice of state sequences of the channel.

Unfortunately, the consideration of stochastic codes creates significant hurdles in further analysis. First, having formulated the problem as a zero-sum game, one is obligated to analyze it using the concept of a \textit{saddle point}~\cite{maschler2013game}. The game is said to admit a saddle point if the minimum loss that the communicating team can incur in the worst case over all jammer strategies (called the \textit{upper value} of the game), equals the maximum loss that the jammer can induce in the worst case over all strategies of the communicating team (called the \textit{lower value} of the game). However, for each fixed action of the jammer, the optimization problem for the communicating team under stochastic codes is necessarily \textit{nonconvex} due to the non-classical information structure \cite{kulkarni2014optimizer}. As a consequence a saddle point need not exist for the resulting zero-sum game. This puts in jeopardy any further game theoretic analysis. 

Second, the construction of a stochastic code entails some intrinsic challenges. To compute an upper bound for the upper value of the game, we require an achievability scheme for the joint source-channel setting. However, constructing a deterministic code for the joint source-channel setting could possibly require determining the existence of a deterministic code of the AVC for the maximum probability of error criterion. The latter problem is unsolved in general and in some cases is equivalent to determining the zero error capacity of a DMC \cite{ahlswede1970note}.

 Our main results shows that despite the nonconvexity above, near-minimax theorems hold for this game in the sense that as the blocklength becomes large, the upper and lower values of the game come arbitrarily close. In particular, we show that there exists a threshold such that if the asymptotic rate is strictly below the threshold, the upper and lower values tend to zero, whereas, for rates strictly above the threshold, the upper and lower values tend to unity. We then consider a refined regime such that the asymptotic rate is exactly equal to the threshold, but allowing a backoff from the threshold that  \textit{changes} with the blocklength along a specific scaling. In this case, the upper and lower values tend to a same constant whose value is neither zero nor one. This shows that a minimax theorem holds even in this novel, refined regime. 

The upper value of the game corresponds to lossy joint source-channel coding over an \textit{arbitrarily varying channel}. The lower value, on the other hand, corresponds to finding a probability distribution on state sequences such that the minimum loss over the resulting (not necessarily memoryless) channel is maximized. Our results are obtained by constructing a converse for the latter problem to lower bound the lower value, and an achievability scheme for the AVC to upper bound the upper value. These bounds yield second order dispersion bounds for the rate of communication which have been of significant interest in the recent past (see, \eg, Polyankskiy \etal~\cite{polyanskiy2010channel}, Kostina and Verd\'{u}~\cite{kostina2013lossy} and their numerous follow ups). In our context, the dispersion bounds yield the aforementioned refined or ``second order'' minimax theorem.

Our achievability scheme exploits the stochastic coding available for joint source-channel coding. 
A separate source-channel coding scheme for the lossless source-channel setting can be constructed by taking the composition of a source code and a channel code constructed independently of each other. Although the scheme is sufficient for the minimax theorems with classical notions of rate, it is not sharp enough to provide us the second order minimax theorem we seek. We develop a joint source-channel coding scheme using ideas from \cite{csiszar2011information} and \cite{kostina2013lossy} to give a stochastic code that yields the desired theorem. 
To construct this achievability scheme we first derive a reduced random code which is a code with a uniform distribution over a smaller number of deterministic codes. Using a deterministic channel code for average probability of error and the above reduced random code, we construct a joint source-channel code which only requires local randomness at the encoder thereby giving a stochastic code. We use the deterministic channel code presented in \cite{kosut2018finite} for our construction.

The lower bound uses the linear programming relaxation method from \cite{jose2017linear}.
Recently, the authors in \cite{jose2017linear} showed that although the point-to-point communication problem (without a jammer) is non-convex in the space of stochastic codes, it nevertheless possesses a hidden convexity. Specifically, they demonstrate that for large blocklengths, the problem can be approximated arbitrarily closely by a linear programming (LP) based relaxation. This suggests that an \textit{approximate} minimax theorem could hold for our problem. Our results validate this intuition.  They show that the LP relaxation is tight even in adversarial settings, adding to the list of cases where the LP relaxation method has yielded tight results~\cite{jose2016linearIT}, \cite{jose2018improvedSW}.

To the best of our knowledge, the formulation we consider here has not been studied before. To be sure, game-theoretic formulation of communication in presence of jammer has been studied previously, but in a somewhat different sense. For instance \cite{borden1985some}, \cite{hegde1989capacity} consider mutual information as the payoff function and prove the existence of saddle point strategies. The control community has studied problem in the continuous alphabet, \eg, the authors in \cite{basar1985complete} consider the problem of communicating a sequence of Gaussian random variables in presence of a jammer and pose a zero-sum game with the mean squared error as the payoff function. In the context of the AVC, the coding theorems for the Gaussian AVC for peak and average power constraints on the input and the jammer were proved in \cite{hughes1987gaussian}. Implicit in their coding theorems is an approximate minimax theorem for the zero-sum game. A zero-sum game formulation similar to our paper is also discussed in \cite[Ch. 12]{csiszar2011information}, where the authors discuss and formulate different instances of the communication over AVC as a zero-sum game. The closest to our setting, is the problem studied in \cite{jose2018game2, jose2018shannon} by the second author of the present paper. There the authors consider only the channel coding problem and the action of the jammer is fixed throughout the transmission, making the upper value problem equivalent to coding for the \textit{compound channel}; moreover,  \cite{jose2018game2, jose2018shannon} do not consider second order minimax theorems.

Lossless joint source-channel coding over an AVC was studied by us in the `conference version' of this paper \cite{vora2019minimax}.
The present paper differs from \cite{vora2019minimax} in the following aspects. We consider lossy joint source-channel setting in this paper unlike in \cite{vora2019minimax} where we considered only the lossless setting. Moreover, we derive dispersion based bounds on the rate, which pave the way for the second order minimax theorem.

This paper is organized as follows. We formulate the problem in Section III. The lower bound is derived in Section IV and the upper bound is derived in Section V. The corresponding asymptotic analysis is done in Section VI and Section VII concludes the paper.

\section{Preliminaries}

\subsection{Notation}
All random variables in this paper are discrete and are defined on an underlying probability space with measure $\mbb{P}$. 
Random variables are represented with uppercase letters $X$ and their instances are denoted by lower case letters $x$; unless otherwise stated these are vectors whose length will be understood from the context. Blackboard letters $\mbb{X},\mbb{Y}$ etc. are used to represent the corresponding single-letter random variable. Caligraphic letters $\mcal{X}, \Yscr$ etc. denote spaces of single-letter random variables. The set of all probability distributions on a space $\mcal{X}$ is denoted by $\mcal{P}(\mcal{X})$ and a particular distribution is represented as $P_{\mbb{X}}$. 

The type of a sequence $x \in \mcal{X}^{n}$ is the empirical distribution $T_x \in \Pscr(\Xscr)$, given by $T_{x}(\bullet) \equiv \frac{|\{ i: x_{i} = \bullet\}|}{n}$. The joint type of $x,y$ is denoted as $T_{x,y}$. The set $\mcal{P}_{n}(\mcal{X}) \subseteq \mcal{P}(\mcal{X})$ denotes the set of all types of sequence in $\mcal{X}^{n}$. The set of sequences  with type $P$ is denoted as $T(P)$.

For any expression $\mcal{B}$, the indicator function $\mbb{I}\{ \mcal{B} \}$ is unity if the expression holds true, otherwise it is zero. The probability of an event $A$ under the measure induced by a distribution $P$ is denoted by $P\{A\} := \sum_{x} \I{x \in A} P(x)$. The variance of a random variable $X \in \mcal{X}$ is denoted as ${\rm Var}(X) := \mbb{E}[ X - \mbb{E}[X]]^{2}$. For any distributions $P_{X}$ and $P_{Y|X}$, we define $(P_{X} \times P_{Y|X})(x,y) := P_{X}(x)P_{Y|X}(y|x)$ and $(P_{X}P_{Y|X})(y) := \sum_{x}P_{X}(x)P_{Y|X}(y|x)$. The complementary Gaussian distribution function is denoted as ${\sf Q}$. Note that this is different than the conditional distribution functions defined as $Q_{X|Y}$ for some random variables $X,Y$. These distribution functions will be defined within the context and used accordingly. The $\exp$ and $\log$ are with respect to base 2. 

\subsection{Arbitrarily Varying Channel}

The arbitrarily varying channel (AVC) was first introduced in \cite{blackwell1960capacities} where the authors consider a channel whose law changes arbitrarily with every transmission. The AVC can be modelled as follows. Consider a family of channels $\mcal{V} := \{P_{\mbb{Y}|\mbb{X},\Theta = \theta} \in \mcal{P}(\mcal{Y}|\mcal{X}), \theta \in \mcal{T} \}$, having common input spaces and output spaces denoted by the finite sets $\mcal{X}$ and $\mcal{Y}$, respectively, where each channel is indexed by the parameter $\theta$, called the state of the channel, drawn from a finite space $\mcal{T}$.
We assume the channel behaviour is memoryless. Thus, the resulting channel can be modeled as the following discrete memoryless AVC. Specifically, the probability of receiving the sequence $y \in \mcal{Y}^{n}$ on sending the input sequence $x$ in $\mcal{X}^{n}$ when the sequence of states is $\theta \in \mcal{T}^{n}$ is given as $  P_{Y|X,\bm{\Theta}}(y|x,\theta) = \prod_{i = 1}^{n} P_{\mbb{Y}|\mbb{X},\Theta}(y_{i}|x_{i},\theta_{i})$.

A \textit{deterministic} channel code is defined as a pair of functions $(f_{\sf C},\varphi_{\sf C})$ given as 
\begin{align}
  f_{\sf C} &: \Wscr \rightarrow \Xscr^n, \varphi_{\sf C} : \Yscr^n \rightarrow \Wscr,
\end{align}
where $\Wscr = \{1,\hdots, M\}$ with $M \in \mbb{N}$. 
 The probability of error for each message $m \in \mcal{W}$ and state sequence $\theta \in \mcal{T}^{n}$ under the deterministic channel code $(f_{\sf C}, \varphi_{\sf C})$ is given as
\begin{align}
  e_{m,\theta}(f_{\sf C},\varphi_{\sf C}) &:=  \sum_{y \in \mcal{Y}^{n}} \mbb{I}\{\varphi(y) \neq m\}  P_{Y|X,\bm{\Theta}}(y |f_{\sf C}(m),\theta).
\end{align}
The maximum and average probability of error over all messages for this code are defined as 
\begin{align}
  e(f_{\sf C},\varphi_{\sf C}) &:= \max_{\theta \in \mcal{T}^{n} } \max_{m \in \{1,\hdots,M\} } e_{m,\theta}(f_{\sf C},\varphi_{\sf C}),  \;\;\;
  \bar{e}(f_{\sf C},\varphi_{\sf C}) := \max_{\theta \in \mcal{T}^{n}} \frac{1}{M}\sum_{m = 1}^{M} e_{m,\theta}(f_{\sf C},\varphi_{\sf C}),
\end{align}
respectively.

Unlike the ordinary discrete memoryless channel (DMC), the capacity of the AVC depends on the error criteria as well as the type of codes used by the encoder and decoder. Further, in some cases, no closed form expression to the capacity of the AVC is known. In particular, it was shown in \cite{ahlswede1970note} that computing the deterministic code capacity of a certain class of AVCs under the maximum probability of error criterion amounts to computing the zero error capacity of a DMC, which is a significantly hard problem. However, the capacity of the AVC is computable for maximum as well as average probability of error criteria if the encoder and decoder are allowed to randomize their actions.


A stochastic code is defined as a pair of conditional distributions $Q_{X|W} \in \mcal{P}(\mcal{X}^{n}| \Wscr), Q_{\wi{W}|Y} \in \mcal{P}(\Wscr|\mcal{Y}^{n})$. 
 The maximum and average probability of error under a stochastic code $(Q_{X|W}, Q_{\wi{W}|Y})$ are defined as
 \begin{align}
  e(Q_{X|W}, Q_{\wi{W}|Y}) &:= \max_{\theta \in \mcal{T}^{n}} \max_{w \in \{1,\hdots,M\}}\sum_{x,y,\wi{w}} \mbb{I}\{\wi{w} \neq w\}   Q_{X|W}(x|w)P_{Y|X,\bm{\Theta}}(y |x,\theta)Q_{\wi{W}|Y}(\wi{w}|y), \non \\
  \bar{e}(Q_{X|W}, Q_{\wi{W}|Y}) &:= \max_{\theta \in \mcal{T}^{n}} \frac{1}{M}\sum_{w=1}^{M}\sum_{x,y,\wi{w}} \mbb{I}\{\wi{w} \neq w\} Q_{X|W}(x|w)  P_{Y|X,\bm{\Theta}}(y |x,\theta)Q_{\wi{W}|Y}(\wi{w}|y). \non
\end{align}

  Let $e$ be the probability of error criterion (maximum or average). Let $R = \frac{\log M}{n}$ be the rate of communication. For $\{\epsilon_{n}\}_{n \geq 1} \rightarrow 0$, the rate $R$ is said to be achievable if there exists a sequence of $(M,n)$ codes of rate $R$ (deterministic or random) such that the probability of error $e \leq \epsilon_{n}$ for sufficiently large $n$. The supremum of all such rates is defined as the capacity of the channel and is denoted as $C$.

In case of the AVC, the capacity $C$ could be zero under certain types of codes and error criteria. We state the conditions under which the capacity of the AVC is positive. An AVC is said to be non-symmetrizable if there does not exist any distribution $P_{\Theta|\mbb{X}} \in \mcal{P}(\mcal{T}| \mcal{X}) $ such that $\forall \; x \in \mcal{X},x' \in \mcal{X}, y \in \mcal{Y}$, we have
  \begin{align}
     \sum_{\theta \in \mcal{T}} P_{\Theta|\mbb{X}}(\theta|x) P_{\mbb{Y}| \mbb{X},\Theta}(y|x',\theta) &=     \sum_{\theta \in \mcal{T}} P_{\Theta|\mbb{X}}(\theta|x') P_{\mbb{Y}| \mbb{X},\Theta}(y|x,\theta). \non 
  \end{align}
 In this paper, it is assumed that the AVC under study is non-symmetrizable. For a non-symmetrizable AVC, the stochastic code capacity for maximum  (and average) probability of error criterion is given as   \cite{ahlswede1978elimination} 
\begin{align}
  C = \max_{P_{\mbb{X}} \in \mcal{P}(\mcal{X})} \min_{q_{\Theta} \in \mcal{P}(\mcal{T})} I(\mbb{X};\mbb{Y}_{q_{\Theta}}), \label{eq:avc-cap}
\end{align}
where $I(\mbb{X};\mbb{Y}_{q_{\Theta}})$ is the mutual information between $\mbb{X}$ and $\mbb{Y}_{q_{\Theta}}$ and $\mbb{Y}_{q_{\Theta}}$ is the output of the averaged channel $(q_{\Theta}P_{\mbb{Y}|\mbb{X},\Theta}) := \sum_{\theta \in \mcal{T}} q_{\Theta}(\theta) P_{\mbb{Y}|\mbb{X},\Theta = \theta}$ when $\mbb{X}$ is the input of the channel. 
For further discussion on the AVC, the reader is referred to (\cite{csiszar2011information}, Ch. 12) and \cite{lapidoth1998reliable}.   

\subsection{Rate distortion function}

Let $\mcal{S}$ be a space of messages and let $d_{\sf S} : \mcal{S} \times \mcal{S} \rightarrow [0,\infty)$ be a distortion function. Consider a problem where it is required to express $k$-length strings from $\mcal{S}^{k}$ in $M$ messages such that the average distortion given by $\mbb{E}[d_{\sf S}(\mbb{S},\wi{\mbb{S}})] \leq \bm{d}$ for a fixed distortion level $\bm{d} > 0$.

A deterministic source code is defined as a pair of functions $(f_{\sf S},\varphi_{\sf S})$ given as 
\begin{align}
  f_{\sf S} &: \Sscr^k \rightarrow \Wscr,\quad  \varphi_{\sf S} :  \Wscr \rightarrow \Sscr^k,
\end{align}
where $\Wscr = \{1 ,\hdots, M\}$ with $M \in \mbb{N}$.

The distortion for the $k$-length sequences is defined as $d(S,\wi{S}) := \frac{1}{k} \sum_{i=1}^{k}d_{\sf S}(\mbb{S}_{i},\wi{\mbb{S}}_{i})  > \bm{d} $. Let the rate be defined as $R = \frac{\log M}{k}$. For a given distortion level $\bm{d}$, the rate $R$ is said to be achievable if there exists a code, such that $\lim_{k \rightarrow \infty} \mbb{E}[d(S,\wi{S})] \leq \bm{d}$. The infimum of all such rates is defined as the rate distortion function and is given as  
\begin{align}
  R(\bm{d}) = \min_{P_{\wi{\mbb{S}}|\mbb{S}}, \mbb{E}[d_{\sf S}(\mbb{S},  \wi{\mbb{S}})]\leq \bm{d}} I(\mbb{S},\wi{\mbb{S}}), \label{eq:rate-dist-func}
\end{align}
where $\mbb{S} \in \mcal{S}$ is distributed as $P_{\mbb{S}} \in \mcal{P}(\mcal{S})$ and $\wi{\mbb{S}} \in \mcal{S}$ is distributed as $P_{\wi{\mbb{S}}|\mbb{S}} \in \mcal{P}(\mcal{S}|\mcal{S})$. More details on the rate-distortion theory can be found in (\cite{cover2012elements}, Ch. 10).

\subsection{Information quantities}
In this section we define few information quantities that will be  used subsequently in the paper. The information density is defined as
\begin{align}
i_{\mbb{X};\mbb{Y}_{q_{\Theta}}}(x;y) &= \log \frac{(q_{\Theta}P_{\mbb{Y}|\mbb{X},\Theta})(y|x)}{(P_{\mbb{X}}q_{\Theta}P_{\mbb{Y}|\mbb{X},\Theta})(y)}, \; x \in \Xscr, y \in \Yscr,   \label{eq:info-den-defn}
\end{align}
where 
\begin{align}
   (q_{\Theta}P_{\mbb{Y}|\mbb{X},\Theta})(y|x) &= \sum_{\theta \in \Tscr} q_{\Theta}(\theta)P_{\mbb{Y}|\mbb{X},\Theta}(y|x,\theta), \;\;
(P_{\mbb{X}}q_{\Theta}P_{\mbb{Y}|\mbb{X},\Theta})(y) = \sum_{x \in \Xscr,\theta \in \Tscr}P_{\mbb{X}}(x)q_{\Theta}(\theta)P_{\mbb{Y}|\mbb{X},\Theta}(y|x,\theta). \non 
\end{align}
The information density for vector valued random variables $(X,Y)$ is denoted as $i_{X;Y_q}$ and is defined analogously. 
The $\bm{d}$-tilted information is defined as 
\begin{align}
j_{{\sf S}}(s,\bm{d}) &= \log \frac{1}{\mbb{E} \left[\exp(\lambda^{*}\bm{d} - \lambda^{*}d(s,\wi{\mbb{S}}))\right]}, \;\;  s \in \Sscr, \label{eq:d-tilted-info}
\end{align}
where the expectation is with respect to the unconditional distribution $P_{\wi{\mbb{S}}^{*}}$ that achieves the minimum in \eqref{eq:rate-dist-func} and  $\lambda^{*} = -R'(\bm{d})$. To define the $\bm{d}$-tilted information for the vector valued random variables $(S,\wi{S})$, we define the following.
\begin{align}
  R_{S}(\bm{d}) = \inf_{P_{\wi{S}|S} : \mbb{E}[d(S,\wi{S})] \leq \bm{d}} I(S;\wi{S}), \label{eq:vec-rate-dist}
\end{align}
where $d(s,\wi{s})$ is the distortion function defined earlier. Further, the $\bm{d}$-tilted information for random variables $S$ is defined as,
\begin{align}
j_{S}(s,\bm{d}) := \log \frac{1}{\mbb{E}\left[\exp(\lambda^{*}\bm{d} - \lambda^{*}d(s,\wi{S}))\right]},  
\end{align}
where the expectation is with respect to the distribution $P_{\wi{S}^*}$ that achieves the infimum in the \eqref{eq:vec-rate-dist} and  $\lambda^{*} = -R_{S}'(\bm{d})$. Further discussion on the $\bm{d}$-tilted information can be found in  \cite{kostina2012fixed}.

We define the following set of capacity achieving distributions
\begin{align}
  \Pi_{\Theta} &= \left\{q_{\Theta} \in \Pscr(\Tscr) : \max_{P_{\mbb{X}}} I(\mbb{X};\mbb{Y}_{q_{\Theta}}) = C \right\},  \;\;
    \Pi_{\mbb{X}} = \left\{ P_{\mbb{X}} \in \Pscr(\Xscr) : \min_{q_{\Theta}} I(\mbb{X};\mbb{Y}_{q_{\Theta}}) = C \right\}.\non
\end{align}
The source and channel dispersions are defined as
 \begin{align}
   V_{{\sf S}} &= {\sf Var}\left(j_{{\sf S}}(\mbb{S},\bm{d})\right), \label{eq:source-disp} \\
   V_{{\sf C}}^{+} &= \min_{P_{\mbb{X}} \in \Pi_{\mbb{X}}} \max_{q_{\Theta} \in  \Pi_{\Theta}}{\sf Var}(i_{\mbb{X};\mbb{Y}_{q_{\Theta}}}(\mbb{X};\mbb{Y})),  \;\;\;  V_{{\sf C}}^{-} =  \max_{q_{\Theta} \in  \Pi_{\Theta}} \min_{P_{\mbb{X}} \in \Pi_{\mbb{X}}} {\sf Var}(i_{\mbb{X};\mbb{Y}_{q_{\Theta}}}(\mbb{X};\mbb{Y})), \label{eq:chan-disp}
 \end{align}
 where $\mbb{S}$ is distributed as $P_{\mbb{S}}$ and $(\mbb{X},\mbb{Y})$ are distributed as $P_{\mbb{X}} \times (q_{\Theta}P_{\mbb{Y}|\mbb{X},\Theta})$. For computing asymptotics, we assume that there exists a unique capacity achieving state distribution $q_{\Theta}^{*} \in \Pi_{\Theta}$. In this case, we have that the above defined channel dispersions are equal $V_{{\sf C}}^{-} = V_{{\sf C}}^{+} =: V_{\sf C}$.

\section{Problem Formulation} \label{sec:form}

 Consider a finite family of channels $\mcal{V} := \{P_{\mbb{Y}|\mbb{X},\Theta = \theta} \in \mcal{P}(\mcal{Y}|\mcal{X}), \theta \in \mcal{T} \}$. Let $\Sscr$ be a finite space. Suppose a random source message $S \in \Sscr^k, k \in \Nbb$, generated i.i.d. according a fixed distribution $P_{\mbb{S}} \in \mcal{P}(\mcal{S})$, 
 is to be communicated over this family of channels where a jammer can choose a channel from the set $\mcal{V}$ for every transmission. An encoder encodes the message $S$ into the channel input string $X \in \mcal{X}^{n}, n \in \mbb{N}$ according to a law $\QXS \in \mcal{P}(\mcal{X}^{n}|\mcal{S}^{k})$. The channel output string $Y \in \mcal{Y}^{n}$ is decoded by the decoder to $\wi{S} \in \mcal{S}^{k}$ according to a law $\QSY \in \mcal{P}(\wi{\mcal{S}}^{k}|\mcal{Y}^{n})$. Together $(\QXS,\QSY)$ is termed as a stochastic code. An error is said to occur if the distortion between the decoded sequence and the source sequence exceeds a predefined level $\bm{d}$, that is, when  $d(S,\wi{S}) := \frac{1}{k} \sum_{i=1}^{k}d_{\sf S}(\mbb{S}_{i},\wi{\mbb{S}}_{i})  > \bm{d} $, where $d_{\sf S}$ is the distortion function defined in Section II-C and $\bm{d} \in [0,\infty)$ is the maximum allowable distortion level. A jammer selects the channels used for transmission by choosing a random state sequence $\bm{\Theta} \in \mcal{T}^{n}$ distributed according to $q \in \mcal{P}(\mcal{T}^{n})$;  $\Theta \in \Tscr$ denotes the single-letterized random variable.  We assume that the encoder and decoder do not know the actions of the jammer and that the jammer also does not have any information about the actions of the encoder and decoder or the source message. 

We assume the channel behaviour is memoryless. Thus, the resulting channel is given by the equation $  P_{Y|X,\bm{\Theta}}(y|x,\theta) = \prod_{i = 1}^{n} P_{\mbb{Y}|\mbb{X},\Theta}(y_{i}|x_{i},\theta_{i})$, which governs the probability of receiving the output sequence $y = (y_{1},\hdots,y_{n})$ when the input sequence is $x = (x_{1},\hdots,x_{n})$ and the state sequence is $\theta = (\theta_{1},\hdots,\theta_{n})$. The rate of communication in this setting is defined as $R = \frac{k}{n}$.

The probability of error is given as 
\begin{align}
 \mbb{P}(d(S,\wi{S}) > \bm{d}) =  &\sum_{s,x,y,\wi{s},\theta}  \mbb{I}\{d(s,\wi{s}) > \bm{d}\}q(\theta) P_S(s)\QXS(x|s)   P_{Y|X,\bm{\Theta}}(y|x,\theta)\QSY(\wi{s}|y). \label{eq:prob-err}
\end{align}
We assume that the encoder and decoder aim to minimize the probability of error by choosing  stochastic codes $(\QXS,\QSY)$ while the jammer tries to maximize it by choosing the distribution $q$. Thus,  for every pair of $(k,n)$, we have a zero-sum game between the encoder-decoder team and the jammer with the probability of error as the payoff. The relevant background on zero-sum games can be found in \cite[Ch. 4]{maschler2013game}.


The minimax or the \textit{upper value} of the game is given by 
$$\problemsmallabb{$\overline{\nu}(k,n)$ }
	{ \QXS, \QSY}
	{\displaystyle  \max_{q} \;\;\mbb{P}(d(S,\wi{S}) > \bm{d}) }
				 {\hspace{-2mm}\begin{array}{r@{\ }c@{\ }l}
				 \QXS \in \mcal{P}(\mcal{X}^{n}|\mcal{S}^{k}), \QSY \in \mcal{P}(\mcal{S}^{k}|\mcal{Y}^{n}), q \in \mcal{P}(\mcal{T}^{n}),
	\end{array}}
$$ 
and the maximin or the \textit{lower value} of the game is given by 
$$\maxproblemsmallbb{$\underline{\nu}(k,n)$ }
	{q}
	{\displaystyle  \min_{ \QXS, \QSY} \mbb{P}(d(S,\wi{S}) > \bm{d}) }
				 {\hspace{-2mm}\begin{array}{r@{\ }c@{\ }l}
				 \QXS \in \mcal{P}(\mcal{X}^{n}|\mcal{S}^{k}), \QSY \in \mcal{P}(\mcal{S}^{k}|\mcal{Y}^{n}), q \in \mcal{P}(\mcal{T}^{n}).
	\end{array}}
$$ 
Clearly, we have that $\overline{\nu}(k,n) \geq \underline{\nu}(k,n)$.


 It can be observed that the minimax problem is a lossy joint source-channel coding problem over an AVC with stochastic codes, since the encoder and decoder search for stochastic codes which minimize the worst case probability of error. Further, it is optimal for the jammer to pick a deterministic sequence of states since the probability of error given by equation \eqref{eq:prob-err} is linear in the distribution $q$.
In a joint source-channel coding problem over a DMC \textit{without} a jammer, asymptotically vanishing probability of error can be achieved for rates below $\frac{C'}{ R(\bm{d})}$, where $C'$ is the capacity of the DMC and $ R(\bm{d})$ is the rate distortion function. Further, the probability of error goes to one for rates above $\frac{C'}{ R(\bm{d})}$ \cite{shannon1959coding}. In this paper, we show that the above result extends to this game, where $\overline{\nu}(k,n)$ and $\underline{\nu}(k,n)$ approach each other as $k,n \rightarrow \infty$, and the value they approach depends on the asymptotic value of the rate $\frac{k}{n}$.


The main results of this paper are the following minimax theorems. Both upper and lower values tend to zero if $  \lim_{k,n \rightarrow \infty} \frac{k}{n} < \frac{C}{ R(\bm{d})}$ as shown in the following result.
\begin{theorem} \label{thm:minimax-tends-0} 
 Consider a sequence $(k,n)$ such that $\lim_{k,n \rightarrow \infty}\frac{k}{n} < \frac{C}{R(\bm{d})}$. 
 Then,
  \begin{align}
  \lim_{k,n \rightarrow \infty} \underline{\nu}(k,n)  =   \lim_{k,n \rightarrow \infty}  \overline{\nu}(k,n) = 0.  
  \end{align}
\end{theorem}

When $\lim_{k,n \rightarrow \infty} \frac{k}{n} > \frac{C}{R(\bm{d})}$, the upper and lower values tend to unity as shown in the following result.
\begin{theorem} \label{thm:minimax-tends-1} 
 Consider a sequence $(k,n)$ such that $\lim_{k,n \rightarrow \infty}\frac{k}{ n} > \frac{C}{R(\bm{d})}$. 
 Then,
  \begin{align}
  \lim_{k,n \rightarrow \infty} \underline{\nu}(k,n)  =   \lim_{k,n \rightarrow \infty}  \overline{\nu}(k,n) = 1.  
  \end{align}
\end{theorem}

We then consider a finer regime. The sequence $(k,n)$ is such that $\frac{k}{n} \rightarrow \frac{C}{R(\bm{d})}$ along a specific scaling. In particular, the sequence is parameterized by $\rho \in \mbb{R}$ given as, 
\begin{align}
  \frac{k}{n} = \frac{C}{R(\bm{d})} + \frac{\rho}{\sqrt{n}}. \label{eq:k-n-seq-L}
\end{align}
For the above sequence, we have the following result.
\begin{theorem}\label{thm:minimax-para-L}
Let $V_{\sf C}^{+} = V_{\sf C}^{-} = V_{\sf C}$.  Then, for the sequence $(k,n)$ chosen as \eqref{eq:k-n-seq-L}, we have
  \begin{align}
    \lim_{k,n \rightarrow \infty} \underline{\nu}(k,n) =     \lim_{k,n \rightarrow \infty} \overline{\nu}(k,n) = {\sf Q} \left(\frac{-\rho R(\bm{d})}{\sqrt{V_{{\sf C}} + \frac{C}{R(\bm{d})} V_{{\sf S}}(\bm{d})}} \right).
  \end{align}
\end{theorem}

The above result states that the upper and lower values of the game tend to an intermediate value between zero and unity, unlike in Theorem \ref{thm:minimax-tends-0} and Theorem \ref{thm:minimax-tends-1}. This non-extremal value is achieved when we take a sequence with the limit as $\frac{C}{R(\bm{d})}$, which is the threshold of reliable communication in the joint source-channel setting. This insight into the finer asymptotics is possible only due to the higher order dispersion bounds that we derive in the latter sections

\section{Lower Bound on the Maximin Value}

We proceed to derive the aforementioned results by computing finite blocklength bounds for $\overline{\nu}(k,n)$ and $\underline{\nu}(k,n)$. In this section, we derive a lower bound on $\underline{\nu}(k,n) $, by relaxing the inner minimization over $(\QXS, \QSY)$ in the maximin problem.
For each $q \in \mcal{P}(\mcal{T}^n)$, the minimization can be written as 
$$\problemsmallawb{SC($q$) }
        { \QXS,\QSY}
	{\displaystyle \sum_{s,x,y,\wi{s}}  \mbb{I}\{d(s,\wi{s}) > \bm{d}\} Q(s,x,y,\wi{s}) }
				 {\begin{array}{r@{\ }c@{\ }l}
                                     Q(s,x,y,\wi{s}) &\equiv& P_{S}(s) \QXS(x|s) P_{Y_{q}|X}(y|x) \\
&& \times \QSY(\wi{s}|y), \\ 
				 \sum_{x} \QXS(x|s) &=& 1 \quad  \forall \; s, \\
 \sum_{\wi{s}} \QSY(\wi{s}|y) &=& 1  \quad \forall \; y, \\
                                    \QXS(x|s),                                    \QSY(\wi{s}|y) &\geq& 0 \quad  \forall \;s,x,y,\wi{s},	\end{array}}
 $$
where $ P_{Y_{q}|X}(y|x) := \sum_{\theta \in \Tscr^n}q(\theta) P_{Y|X,\bm{\Theta}}(y|x,\theta)$.

The above problem is non-convex in the space of the distributions $(\QXS, \QSY)$ \cite{kulkarni2015optimizer}. A particular line of approach for such problems is to derive a convex relaxation by containing the non-convex feasible region within a convex set. We consider a linear programming relaxation presented in \cite{jose2017linear} derived by a lift-and-project like method.

In this method, we linearize the objective and the constraints by replacing the bi-product terms \newline $\QXS(x|s) \QSY(\wi{s}|y)$ with an auxiliary variable $V(s,x,y,\wi{s})$. Further, we append constraints which imply that the variable $V(s,x,y,\wi{s}) \equiv \QXS(x|s) \QSY(\wi{s}|y)$. Finally, to derive a relaxation of the problem SC($q$), we drop the requirement that $V(s,x,y,\wi{s}) \equiv \QXS(x|s) \QSY(\wi{s}|y)$. Effectively, we approximate the non-convex feasible region of the problem SC($q$) by a polytope, thereby \textit{lifting} the problem into a higher dimensional space. This exercise results in a linear program which is given as 
$$ \problemsmallawb{LP($q$)}
	{ \begin{subarray}{1}
           \QXS,\QSY,\\
           \qquad V
         \end{subarray}}
	{\displaystyle \sum_{s,x,y,\wi{s}} \mbb{I}\{d(s,\wi{s}) > \bm{d}\} \bar{Q}(s,x,y,\wi{s})}
        {\begin{array}{r@{\ }c@{\ }l}
           \bar{Q}(s,x,y,\wi{s}) &\equiv& P_{S}(s)P_{Y_{q}|X}(y|x) \\
           && \quad \times V(s,x,y,\wi{s}) \\
        \sum_{x} \QXS(x|s)&=& 1 \hspace{0.05cm} : \gamma^{{\sf S}}_{q}(s) \qquad  \hspace{0.08cm} \forall s,\\
    \sum_{\wi{s}} \QSY(\wi{s}|y) &=&1 \hspace{0.05cm} :\gamma^{{\sf C}}_{q}(y) \qquad \hspace{0.08cm}\forall y,\\
 \sum_{x} V(s,x,y,\wi{s})-\QSY(\wi{s}|y)&=&0  \hspace{0.05cm} :\lambda^{{\sf S}}_{q}(s,\wi{s},y)  \hspace{0.0cm}  \forall s,\wi{s},y,\\ 
 \sum_{\wi{s}}V(s,x,y,\wi{s})-\QXS(x|s)&=&0  \hspace{0.05cm} :\lambda^{{\sf C}}_{q}(x,s,y)\hspace{0.0cm} \forall x,s,y,\\
\QXS(x|s), \QSY(\wi{s}|y)  &\geq& 0   \quad \hspace{1cm}\forall s,x,y,\wi{s}, \\
 V(s,x,y,\wi{s})&\geq& 0 \hspace{0.05cm}   \qquad  \hspace{0.6cm} \forall s,x,y,\wi{s},
         \end{array}}  $$
       where the functions  $\gamma^{{\sf S}}_{q}:\mcal{S}^{k} \rightarrow \Real$, $\gamma^{{\sf C}}_{q}:\mcal{Y}^{n} \rightarrow \Real$, $\lambda^{{\sf S}}_{q}:\mcal{S}^{k} \times \mcal{S}^{k} \times \mcal{Y}^{n} \rightarrow \Real$ and  $\lambda^{{\sf C}}_{q}: \mcal{S}^{k} \times \mcal{X}^{n} \times \mcal{Y}^{n} \rightarrow \Real$  are Lagrange multipliers. 
       
       Any feasible point of the problem LP($q$) is given by the tuple $(\QXS,\QSY,V)$ and the corresponding feasible point of SC($q$) is derived by projecting this point onto the space of $(\QXS,\QSY)$. Successively repeating the exercise results in increasingly tighter convex relations of the original problem. Further details on the lift-and-project method can be found in (\cite{conforti2014integer}, ch. 5).

       Clearly, we have $$\OPT(\SC(q)) \geq \OPT(\LP(q)), \qquad \forall \; q \in \mcal{P}(\mcal{T}^{n}).$$ Now we derive the dual problem of LP($q$) and using weak duality, we bound optimal value of LP($q$) thereby bounding $\OPT$(SC($q$)). The corresponding dual program is given as follows.
$$
\maxproblemsmallwb{DP($q$) \;}
	{\gamma^{{\sf S}}_{q},\gamma^{{\sf C}}_{q},\lambda^{{\sf S}}_{q},\lambda^{{\sf C}}_{q}}
	{\displaystyle \sum_{s}\gamma^{{\sf S}}_{q}(s)+\sum_{y}\gamma^{{\sf C}}_{q}(y)}
				 {\begin{array}{r@{\ }c@{\ }l}
				 				 \gamma^{{\sf S}}_{q}(s)- \sum_y \lambda^{{\sf C}}_{q}(x,s,y)
                                    &\leq& 0 \quad \hspace{1.3cm} \forall x,s, \hspace{0.8cm} \rm{(I)} \\
\gamma^{{\sf C}}_{q}(y)- \sum_s \lambda^{{\sf S}}_{q}(s,\wi{s},y)
                                    &\leq& 0  \quad \hspace{1.3cm} \forall \shat,y, \hspace{0.8cm} \rm{(II)}\\
\lambda^{{\sf S}}_{q}(s,\wi{s},y)+\lambda^{{\sf C}}_{q}(x,s,y)
                                    &\leq& \Pi(s,x,y,\wi{s})\hspace{0.15cm}\forall s,x,y,\wi{s}, \hspace{0.09cm} {\rm (III)}
 	\end{array}}
      $$ where $\Pi(s,x,y,\wi{s}) \equiv \mbb{I}\{d(s,\wi{s}) > \bm{d}\} P_S(s) P_{Y_{q}|X}(y|x)$. From weak duality it follows that  $\OPT({\rm SC }(q)) \geq \OPT({\rm LP}(q)) = \OPT({\rm DP}(q))\; \forall \; q$ and from \cite{jose2018game2} we get 
      \begin{align}
 \underline{\nu}(k,n) =  \max_{q} \OPT({\rm SC}(q)) \geq   \max_{q}  \OPT({\rm LP}(q)) =  \max_{q} \OPT({\rm DP}(q)) \geq \max_{q} \;{\rm FEAS}({\rm DP}(q)), \label{eq:main-ineq}
      \end{align}
where ${\rm FEAS}({\rm DP}(q))$ is the objective function of ${\rm DP}(q)$ evaluated at a feasible point. Thus, to compute a lower bound on the minimax, as well as, the maximin value of the zero-sum game, it is sufficient to derive a feasible solution of the ${\rm DP}(q)$. The following Theorem gives one such construction and computes the dual cost of ${\rm DP}(q)$ for the above feasible solution to get a lower bound for the maximin value.

\begin{theorem} \label{thm:maximin-val-bnd}
  The value $ \underline{\nu}(k,n)$ is lower bounded as
    \begin{align}
\underline{\nu}(k,n) \geq \max_{q}  \OPT({\rm DP}(q)) \geq \max_{q,P_{\overline{Y}_{q}}, {\sf U}} \; \sup_{\gamma > 0}  \left[  \sum_{s}P_{S}(s) \right. & \min_{x} \left[ \vphantom{ \sum_{u=1}^{{\sf U}}} \mbb{P}\left( j_{S}(s,\bm{d}) - i_{X;\overline{Y}_{q}|U}(x;Y|U) \leq \gamma \right)\right. \non \\ 
                    &  + \exp ( j_{S}(s,\bm{d})-\gamma) \sum_{u=1}^{{\sf U}} \sum_{y}P_{U|X}(u|x) P_{\overline{Y}_{q}|U}(y|u) \non \\
&\times \left. \left. \mbb{I}\left\{  j_{S}(s,\bm{d}) - i_{X;\overline{Y}_{q}|U}(x;y|u) > \gamma \right\}  \vphantom{ \sum_{u=1}^{{\sf U}}} \right]- \frac{{\sf U}}{\exp(\gamma)} \right],        \label{eq:dp-bound}
    \end{align}
where $U \in \Uscr := \{1,\hdots,{\sf U}\}$, ${\sf U} \in \Nbb$, $P_{\overline{Y}_{q}} \in \Pscr(\Yscr^{n})$,  $i_{X;\overline{Y}_{q}|U}(x;y|u) =  \log \frac{P_{Y_{q}|X,U}(y|x,u)}{P_{\overline{Y}_{q}|U}(y|u)}$ and $j_{S}(s,\bm{d})$ is the $\bm{d}$-tilted information as defined in Section II-D.
\end{theorem}
    \begin{proof}
Define a random variable $U$ taking values in $\Uscr := \{1,\hdots,{\sf U}\}$ such that 
\begin{align}
P_{Y_{q}|X}(y|x) = \sum_{u=1}^{{\sf U}} P_{U|X}(u|x)P_{Y_{q}|X,U}(y|x,u).
\end{align}
Consider the following dual variables for ${\rm DP}(q)$
    \begin{align}
      \lambda^{{\sf S}}_{q}(s,\wi{s},y) &\equiv -\mbb{I}\{d(s,\wi{s}) \leq \bm{d}\}  \frac{P_{S}(s) \sum_{u=1}^{{\sf U}}P_{\overline{Y}_{q}|U}(y|u)}{\exp ( \gamma - j_{S}(s,\bm{d}))}, \non \\
      \lambda^{{\sf C}}_{q}(x,s,y) &\equiv P_{S}(s) \sum_{u=1}^{{\sf U}}P_{U|X}(u|x) \min \left\{P_{Y_{q}|X,U}(y|x,u), \frac{P_{\overline{Y}_{q}|U}(y|u)}{\exp ( \gamma - j_{S}(s,\bm{d}) )}\right\}, \non \\ 
   \gamma^{{\sf S}}_{q}(s) &\equiv \min_{x} \sum_{y} \lambda^{{\sf C}}_{q}(x,s,y), \gamma^{{\sf C}}_{q}(y) \equiv -\exp(-\gamma) \sum_{u=1}^{{\sf U}} P_{\overline{Y}_{q}|U}(y|u),\non 
  \end{align}
where $\gamma > 0$,  $ P_{\overline{Y}_{q}|U}(y|u) :=  \sum_{\theta} q(\theta) P_{\overline{Y}|\bm{\Theta},U}(y|\theta,u)$ and $P_{\overline{Y}|\bm{\Theta},U}$ is any distribution in $\Pscr(\Yscr^n|\Tscr^n,\Uscr).$

From the proof of  Theorem 5.3 in \cite{jose2017linear}, it follows that the above choice of dual variables are feasible for the ${\rm DP}(q)$. Further, the dual cost is given as 
      \begin{align}
        &\sum_{s} \min_{x} \sum_{y} \lambda^{{\sf C}}_{q}(x,s,y) + \sum_{y}  -\exp(-\gamma) \sum_{u=1}^{{\sf U}} P_{\overline{Y}_{q}|U}(y|u) \non \\
        &\geq \sum_{s} \min_{x} \sum_{y}P_{S}(s) \sum_{u=1}^{{\sf U}}P_{U|X}(u|x)  \min \left\{P_{Y_{q}|X,U}(y|x,u), \frac{P_{\overline{Y}_{q}|U}(y|u)}{\exp ( \gamma - j_{S}(s,\bm{d}) )}\right\}  - \frac{\sum_{y} \sum_{u=1}^{{\sf U}} P_{\overline{Y}_{q}|U}(y|u)}{\exp(\gamma)} \non \\ 
        &= \sum_{s} P_{S}(s)\min_{x} \left[ \sum_{u=1}^{{\sf U}}\sum_{y} P_{U|X}(u|x)P_{Y_{q}|X ,U}(y|x,u) \mbb{I}\left\{ \frac{P_{Y_{q}|X,U}(y|x,u)}{P_{\overline{Y}_{q}|U}(y|u)} \leq \frac{\exp( j_{S}(s,\bm{d}))}{ \exp(\gamma) }\right\}\right. \non \\ 
                    &  + \frac{1}{\exp ( \gamma - j_{S}(s,\bm{d}) )} \sum_{u=1}^{{\sf U}} \sum_{y}P_{U|X}(u|x) P_{\overline{Y}_{q}|U}(y|u) \left. \mbb{I}\left\{ \frac{P_{Y_{q}|X,U}(y|x,u)}{P_{\overline{Y}_{q}|U}(y|u)} > \frac{\exp( j_{S}(s,\bm{d}))}{ \exp(\gamma) } \right\} \right]-  \frac{{\sf U}}{\exp(\gamma)}. \non 
        \end{align}
Taking $i_{X;\overline{Y}_{q}|U}(x;y|u) = \log \frac{P_{Y_{q}|X,U}(y|x,u)}{P_{\overline{Y}_{q}|U}(y|u)}$, we get
\begin{align}
          &\sum_{s} \min_{x} \sum_{y} \lambda^{{\sf C}}_{q}(x,s,y) + \sum_{y}  -\exp(-\gamma) \sum_{u=1}^{{\sf U}} P_{\overline{Y}_{q}|U}(y|u) \non \\
        &= \sum_{s} P_{S}(s)\min_{x} \left[ \vphantom{ \sum_{u=1}^{{\sf U}}} \mbb{P}\left( i_{X;\overline{Y}_{q}|U}(x;Y|U) - j_{S}(s,\bm{d}) \leq -\gamma \right)\right.   \non \\
        &+ \exp ( j_{S}(s,\bm{d})-\gamma) \sum_{u=1}^{{\sf U}} \sum_{y}P_{U|X}(u|x) P_{\overline{Y}_{q}|U}(y|u)  \left. \mbb{I}\left\{ i_{X;\overline{Y}_{q}|U}(x;y|u) - j_{S}(s,\bm{d}) > -\gamma \right\}  \vphantom{ \sum_{u=1}^{{\sf U}}} \right]-  \frac{{\sf U}}{\exp(\gamma)}. \non
\end{align}
Taking supremum over $\gamma$ and maximum over ${\sf U}$, $ P_{\overline{Y}_{q}}$ and  $q$ in the above equation, we get the expression on the RHS of \eqref{eq:dp-bound}. The required bound follows from the equation \eqref{eq:main-ineq}. 
    \end{proof}

The linear programming relaxation gives a lower bound on the upper and lower values of the game. We now derive an upper bound by constructing an achievability code for the joint source-channel setup.

\section{Upper Bound on the Minimax Value}

To construct an upper bound on the minimax value $ \overline{\nu}(k,n)$, we construct a stochastic joint source-channel code. Recall that a stochastic joint source-channel code is a pair of distributions $(\QXS, \QSY)$ with $\QXS \in \Pscr(\Xscr^n|\Sscr^k)$ and $\QSY \in \Pscr(\Sscr^k|\Yscr^n)$. 

In this paper, we consider a stochastic code involving a stochastic encoder and a deterministic decoder. To construct the stochastic code, we consider a random joint source-channel code and a deterministic channel code. Using the idea of random code reduction, we derive another random code with uniform distribution over a smaller number of codes. The encoder randomly chooses a code from this ensemble to communicate the source message. Further, the encoder uses the deterministic code to communicate the index across the channel to the decoder. The decoder then decodes the index with the deterministic code and then uses the code corresponding to the index to decode the source message.

To construct the stochastic joint source-channel code, we first consider a deterministic channel code which gives a guarantee on the average probability of error.

\subsection{Deterministic channel code for average probability of error}

In this section, we first consider a deterministic channel code which is independent of the source code. Let $(f_{{\sf c}}, \varphi_{{\sf c}})$ be a channel code as defined in Section II-B. 

We have the following result from Theorem 1 in \cite{kosut2018finite} which provides for the existence of a deterministic code for average probability of error.
\begin{theorem} \label{thm:kosut-kliewer-achieve}
 Let $P_{X} \in \mcal{P}(\mcal{X}^{n})$ and  $Z(x,\bar{x},y) \in \{0,1\}$ be a function such that
  \begin{align}
    Z(x,\bar{x},y) Z(\bar{x},x,y) = 0 \quad \forall \; x \in \mcal{X}^{n}, \bar{x} \in \mcal{X}^{n}, y \in \mcal{Y}^{n} \label{eq:Z-func-cond}
  \end{align}
  and $\mcal{A} \subset \mcal{X}^{n} \times \mcal{Y}^{n}$ is a typical set. Let $(X,\bar{X}, Y) \sim P_{X} \times P_{X}\times P_{Y|X,\bm{\Theta} }$. Then, there exists a deterministic channel code $(f_{{\sf c}}, \varphi_{{\sf c}})$, such that
  \begin{align}
\bar{e}(f_{{\sf c}}, \varphi_{{\sf c}}) 
&\leq \sqrt{\frac{2\ln (3|\mcal{T}|^{n})}{M} } + \min_{P_{X}} \max_{\theta \in \mcal{T}^{n}} \left( \vphantom{\max_{\bar{x}} } \mbb{P}((X,Y) \notin \mcal{A} | \theta ) \right.  +  2 \log 3|\mcal{T}|^{n}\max_{\bar{x} \in \mcal{X}^{n}}   \mbb{P}(Z(X,\bar{x},Y) = 0, (X,Y) \in \mcal{A}  | \theta ) \non \\
&\left. + 2M\log e \; \mbb{P}(Z(X,\bar{X},Y) = 0, (X,Y) \in \mcal{A} | \theta ) \vphantom{\max_{\bar{x}} } \right). \label{eq:kosut-kliewer-achieve}
  \end{align}
\end{theorem}

We thus have a deterministic channel code that gives an upper bound on the average probability of error. We now construct a random joint source-channel code.

\subsection{Random joint source-channel code}

\begin{figure}
\begin{center}
\begin{tikzpicture}[node distance=2.5cm,auto,>=latex']
  \node [int] (enc1) {$f_{{\sf S}}$}; \node [int] (enc2) [right
  of=enc1, node distance=1.5cm ] {$f_{{\sf C}}$}; \node [int] (cha)
  [right of=enc2, node distance=1.9cm] {$P_{Y|X,\bm{\Theta}}$};; \node
  [int] (dec2) [right of=cha, node distance=1.9cm] {$\varphi_{{\sf C}}$};
  \node [int] (dec1) [right of=dec2, node distance=1.5cm]
  {$\varphi_{{\sf S}}$};
  \node (msg) [right of=enc1, node distance=0.8cm] {}; \node (msg2) [above of=msg, node distance=0.7cm]
  {};
  \node (msg3) [right of=dec2, node distance=0.8cm] {};
  \node (msg4) [above of=msg3, node distance=0.7cm] {};
    \node (jam) [pin={[init]above:$\bm{\Theta}$}, above of= cha, node distance=0.23cm]  {};
    \node (start) [left of=enc1,node distance=0.9cm, coordinate] {};
    \node (end) [right of=dec1, node distance=1cm] {};
    \path[->] (start) edge node {$S$} (enc1);
    \path[->] (enc1) edge node {$W$} (enc2);
    \path[->] (enc2) edge node {$X$} (cha);
    \path[->] (cha) edge node {$Y$} (dec2);
    \path[->] (dec2) edge node {$\wi{W}$} (dec1);
    \path[->] (dec1) edge node {$\wi{S}$} (end);
\end{tikzpicture}
\end{center}  
  \caption{Source-Channel Coding Setup}
  \label{fig:source-channel}
\end{figure}
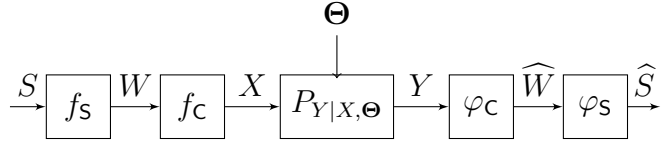

The joint source-channel coding setup is given in Figure \ref{fig:source-channel}. Let the output random variable of the source encoder be $W \in  \Wscr = \{1,\hdots, M\},$ where $ M \in \mbb{N}$, and the input random variable of the source decoder be $\wi{W} \in  \Wscr$. The source code is a pair of functions $(f_{{\sf S}}, \varphi_{{\sf S}})$ as defined in Section II-C. The channel code is a pair of functions $(f_{{\sf C}}, \varphi_{{\sf C}})$ as defined in Section II-B. The composition of the two codes gives the joint source-channel code $(f,\varphi)$  defined as
\begin{align}
  f &:= (f_{{\sf C}} \circ f_{{\sf S}}) : \Sscr^{k} \rightarrow \Xscr^{n}, \varphi := (\varphi_{{\sf S}} \circ \varphi_{{\sf C}}) : \Yscr^{n} \rightarrow \Sscr^{k}. \non 
\end{align}
A random joint source-channel code is a pair of random variables $(F,\Phi)$ taking values in the set $ \{ (f,\varphi)  | f : \Sscr^{k} \rightarrow \Xscr^{n}, \varphi  : \Yscr^{n} \rightarrow \Sscr^{k} \}$.
The probability of error for the deterministic joint source-channel code is given as $ e_{\bm{d},\theta}(f,\varphi) :=  \sum_{s,y} \mbb{I}\left\{ d(s,\varphi(y)) > \bm{d} \right\}P_{S}(s)  P_{Y|X,\bm{\Theta}}(y|f(s),\theta)$.
Correspondingly, the error for the random joint source-channel code $(F,\Phi) \sim \psi$  is given as 
\begin{align}
 e_{\bm{d},\theta}(\psi) &:= \mbb{E} \left[e_{{\bm{d},\theta}}(F,\Phi) \right], 
  e_{\bm{d}}(\psi) := \max_{\theta \in \mcal{T}^{n}} \; \mbb{E} \left[  e_{\bm{d},\theta}(\psi) \right].
\end{align}

We now construct a random joint source-channel code which is similar to the construction of a deterministic code given in \cite{kostina2013lossy} for a joint source-channel setting with a non-adversarial DMC.

\begin{theorem} \label{thm:joint-random-code-achieve}
There exists a random code $\psi$ such that the probability of error is upper bounded as
  \begin{align}
&   e_{\bm{d}}(\psi) \leq \min_{P_{L|S}} \max_{\theta \in \mcal{T}^{n}} \left( \mbb{E} \left[ \exp\left( -\left| i_{X;Y_{q^*}}(X;Y) - \log L \right|^{+} \right) |  \theta \right]  + \mbb{E}\left[ (1 - P_{\wi{S}}\left(\mcal{B}_{\bm{d}}(S))\right)^{L} \right]\right),
  \end{align}
where $(S,L,\wi{S},X,Y) \sim P_{S} \times P_{L|S} \times P_{\wi{S}} \times P_{X} \times P_{Y|X,\bm{\Theta} = \theta}$, 
\begin{align}
 i_{X;Y_{q^*}}(x;y) = \log \frac{(q^*P_{Y|X,\bm{\Theta}})(y|x)}{(P_{X}q^*P_{Y|X,\bm{\Theta}})(y)}, 
\end{align}
where $(X,Y) \sim P_{X} \times P_{Y|X,\bm{\Theta} = \theta}$, $q^*(\theta) = \prod_{i=1}^{n}q_{\Theta}^{*}(\theta_i), q_{\Theta}^{*} \in \Pi_{\Theta}$ and $\mcal{B}_{\bm{d}}(s) := \{\wi{s} \in \Sscr^{k} : d(s,\wi{s}) \leq \bm{d}\}$.  
\end{theorem}
\begin{proof}
We construct a random code $\psi$ as follows. First, we construct a random source code $(F_{{\sf S}}, \Phi_{{\sf S}})$. Then, for a particular instance of the source code $(f_{\sf S}, \varphi_{\sf S})$, we construct a random channel code $(F_{{\sf C}}, \Phi_{{\sf C}})$.

  \textit{Source encoder - }  
Generate $M$ codewords from the set $\Sscr^{k}$ each according to $P_{\wi{S}} \in \Pscr(\Sscr^{k})$. Let the $i^{th}$ codeword be denoted as $\wi{S}_{i}$. 
Further, for a given source message $s \in \Sscr^{k}$, generate a random variable $L \sim P_{L|S = s} \in \Pscr(\mbb{N}|\Sscr^{k})$ such that $L \leq M$. The encoder encodes the source message $s$ as
\begin{align}
 F_{{\sf S}}(s) = \left\{
      \begin{array}{c l}
        \min\{m,L\} & d(s,\wi{S}_{m}) \leq \bm{d} < \min_{i < m}d(s,\wi{S}_{i}) \\
        L & \bm{d} < \min_{i=1,\hdots,L}d(s,\wi{S}_{i}),
      \end{array}
\right.
\end{align}
where $d$ is the distortion function and $\bm{d}$ is the maximum distortion level as defined in Section II-C. 

 \textit{Source decoder - } The source decoder decodes the output $m$ of the channel decoder as $\Phi_{\sf S}(m) = \wi{S}_{m}$. 

We now construct a random channel code for a given instance of source code $(F_{{\sf S}}, \Phi_{{\sf S}}) = (f_{{\sf S}}, \varphi_{{\sf S}})$.

\textit{Channel encoder - } Define $M$ random variables $\{X_{m}\}_{m=1}^{M}$ each distributed according to $P_{X} \in \Pscr(\Xscr^{n})$. For a given output $m$ of the source encoder, the channel encoder encodes the output as $F_{{\sf C}}(m) = X_{m}$.

\textit{Channel decoder -} For the channel decoding function, we define a random variable $U \in \{1,\hdots,M+1\}$ as follows
\begin{align}
 U = \left\{
      \begin{array}{c l}
        f_{{\sf S}}(S) & d(S,\varphi_{{\sf S}}\circ f_{{\sf S}}(S)) \leq \bm{d}  \\
        M+1 & {\rm else}.
      \end{array}
\right.
\end{align}
Let $P_{Y_{q^*}|X}(y|x) := \sum_{\theta \in \Tscr^{n}}q^*(\theta)P_{Y|X,\bm{\Theta}}(y|x,\theta)$, where $q^*(\theta) = \prod_{i=1}^{n}q_{\Theta}^{*}(\theta_i), q_{\Theta}^{*} \in \Pi_{\Theta}$. For the observed output $y$ of the channel, the channel decoder decodes the string $y$ as
\begin{align}
  \Phi_{{\sf C}}(y) = m \in \arg \max_{j \in  \{1,\hdots,M\}} P_{U|\wi{S}^{M}}(j|\wi{s}^{M})P_{Y_{q^*}|X}(y|X_{j}),
\end{align}
where $P_{U|\wi{S}^{M}} \in \Pscr \left(\{1,\hdots,M+1\}|(\Sscr^{k})^{M}\right)$.

Let $\psi$ be the distribution induced by the source codebook distribution $P_{\wi{S}}$ and the channel codebook distribution $P_{X}$ on the set of joint source-channel codes $\{f,\varphi\}$. We now compute the probability of error under this random code $\psi$ for a fixed $\theta \in \Tscr^{n}$. Let the source and the channel codebooks be $\wi{s}^{M} = (\wi{s}_{1} ,\hdots, \wi{s}_{M}) \in \Sscr^{k}$ and $x^{M} = (x_1,\hdots,x_M) \in \Xscr^{n}$. The probability of error is bounded above as
\begin{align}
\mbb{P}\left(d(S,\wi{S}) > \bm{d} | \wi{s}^{M},x^{M}, \theta\right) &\leq   \mbb{P}\left( d(S,\varphi_{{\sf S}}\circ f_{{\sf S}}(S)) > \bm{d} | \wi{s}^{M}\right) \non \\ 
&+  \mbb{P} \left(\varphi_{{\sf C}}(Y) \neq f_{{\sf S}}(S) | \wi{s}^{M}, x^{M}, \theta, d(S,\varphi_{{\sf S}}\circ f_{{\sf S}}(S)) \leq \bm{d} \right), \label{eq:porb-err-jscc}
\end{align}
where the first term is the source coding error and the second term is the channel decoding error when there is no error at the source encoder. The source coding error is given as
 $   \mbb{P}\left( d(S,\varphi_{{\sf S}}\circ f_{{\sf S}}(S)) > \bm{d} | \wi{s}^{M}\right) =  \mbb{P}\left( U > L | \wi{s}^{M}\right)$ 
since $\{ d(S,\varphi_{{\sf S}}\circ f_{{\sf S}}(S)) > \bm{d} \} = \{ U > L\}$. The channel decoding error can be further computed as follows.
\begin{align}
&     \mbb{P} \left(\varphi_{{\sf C}}(Y) \neq f_{{\sf S}}(S) | \wi{s}^{M}, x^{M}, \theta, d(S,\varphi_{{\sf S}}\circ f_{{\sf S}}(S)) \leq \bm{d} \right)  \non \\
     &=  \mbb{P} \left( \varphi_{{\sf C}}(Y) \neq U | \wi{s}^{M}, x^{M}, \theta \right) \label{eq:fs-eq-U} \\
  &= \sum_{m=1}^{M} P_{U|\wi{S}^{M}}(m|\wi{s}^{M}) \mbb{P}\left( \varphi_{{\sf C}}(Y) \neq m |  x^{M}, \theta \right) \non \\
  &=   \sum_{m=1}^{M}P_{U|\wi{S}^{M}}(m|\wi{s}^{M}) \mbb{P}\left( \bigcup_{m'} \frac{P_{U|\wi{S}^{M}}(m'|\wi{s}^{M})P_{Y_{q^*}|X}(Y|X_{m'})}{P_{U|\wi{S}^{M}}(m|\wi{s}^{M})P_{Y_{q^*}|X}(Y|X_{m})} \geq 1 | x^{M}, \theta\right), 
\end{align}
where \eqref{eq:fs-eq-U} follows since for $d(S,\varphi_{{\sf S}}\circ f_{{\sf S}}(S)) \leq \bm{d} $ implies that $f_{{\sf S}}(S) = U$.

Averaging over codebooks $(\wi{s}^M,x^M)$, we get that
\begin{align}
  &\sum_{\wi{s}^M,x^M}P_{\wi{S}^M}(\wi{s}^M)P_{X^M}(x^M) \mbb{P}\left(d(S,\wi{S}) > \bm{d} | \wi{s}^{M},x^{M}, \theta\right) = \mbb{P}\left(d(S,\wi{S}) > \bm{d} | \theta\right)   = e_{\bm{d},\theta}(\psi),
\end{align}
where $P_{\wi{S}^M}(\wi{s}^M) = \prod_{i=1}^{M}P_{\wi{S}}(\wi{s}_{i})$, $\wi{s}_{i} \in \Sscr^{k}$ and $P_{X^M}(x^M) = \prod_{i=1}^{M}P_{X}(x_{i})$, $x_{i} \in \Xscr^{n}$. $\mbb{P}\left(d(S,\wi{S}) > \bm{d} | \theta\right)   = e_{\bm{d},\theta}(\psi)$ follows since the random code is induced by the distributions $P_{X^M} $  and $ P_{\wi{S}^M}$.

Substituting the above in \eqref{eq:porb-err-jscc}, we get
\begin{align}
& e_{\bm{d},\theta}(\psi) \leq  \mbb{P}(U > L) +     \sum_{\wi{s}^{M}}P_{\wi{S}^{M}}(\wi{s}^{M}) \sum_{m=1}^{M}P_{U|\wi{S}^{M}}(m|\wi{s}^{M})  \mbb{P}\left( \;\bigcup_{m'=1}^{M} \frac{P_{U|\wi{S}^{M}}(m'|\wi{s}^{M})P_{Y_{q^*}|X}(Y|X_{m'})}{P_{U|\wi{S}^{M}}(m|\wi{s}^{M})P_{Y_{q^*}|X}(Y|X_{m})} \geq 1 | \theta\right). \label{eq:perr-sum-chan-source-err}
\end{align}
 Following the line of arguments given in Theorem 7 from \cite{kostina2013lossy}, we get 
 \begin{align}
&  \sum_{\wi{s}^{M}}P_{\wi{S}^{M}}(\wi{s}^{M}) \sum_{m=1}^{M}P_{U|\wi{S}^{M}}(m|\wi{s}^{M})  \mbb{P}\left( \;\bigcup_{m'=1}^{M} \frac{P_{U|\wi{S}^{M}}(m'|\wi{s}^{M})P_{Y_{q^*}|X}(Y|X_{m'})}{P_{U|\wi{S}^{M}}(m|\wi{s}^{M})P_{Y_{q^*}|X}(Y|X_{m})} \geq 1 |  \theta\right) \non \\
&\leq \mbb{E}\left[\exp\left(-|i_{X;Y_{q^*}}(X;Y) - \log L|^{+}\right) | \theta \right]
 \end{align}
and $  \mbb{P}(U > L) = \mbb{E}\left[\left(1-\mbb{P}(\Bscr_{\bm{d}}(S))\right)^{L}\right]$.
Substituting in \eqref{eq:perr-sum-chan-source-err}, we get 
\begin{align}
e_{\bm{d},\theta}(\psi) &\leq \mbb{E}\left[\exp\left(-|i_{X;Y_{q^*}}(X;Y) - \log L|^{+} | \theta\right)\right]  +  \mbb{E}\left[\left(1- P_{\wi{S}}(\Bscr_{\bm{d}}(S))\right)^{L}\right].
\end{align}
Taking minimum over the distributions $P_{L|S}$ and taking the maximum over the states $\theta$, we get the required bound.
\end{proof}

We further state a weaker bound by choosing  $L = \floor{\frac{\gamma}{P_{\wi{S}}(\Bscr_{\bm{d}}(S))}}$ where $\gamma > 0$ is chosen arbitrarily according to Theorem 8 in \cite{kostina2013lossy}.
\begin{theorem} \label{thm:joint-random-code-achieve-weak}
  There exists a random joint source-channel code $\psi$ such that 
  \begin{align}
&   e_{\bm{d}}(\psi) \leq  \inf_{\gamma > 0}\max_{\theta \in \mcal{T}^{n}}  \left(  \mbb{E} \left[ \exp\left( -\left| i_{X;Y_{q^*}}(X;Y) -  \frac{\gamma}{P_{\wi{S}}(\Bscr_{\bm{d}}(S))} \right|^{+} \right) | \theta \right]  - e^{1-\gamma} \vphantom{\frac{\gamma}{P_{\wi{S}}(\Bscr_{\bm{d}}(S))}} \right),
  \end{align}
where $(S,\wi{S},X,Y)$ are as defined in Theorem \ref{thm:joint-random-code-achieve}.
\end{theorem}  
Using the deterministic channel code and the random joint source-channel code, we now construct a stochastic joint source-channel code.

\subsection{Stochastic source-channel code and upper bound}

In this section, we present an approach to construct a stochastic joint source-channel code. For that, we consider a lemma which is a version of the random code reduction lemma given in  \cite[Ch. 4]{csiszar2011information}.
\begin{lemma} \label{lem:ran-code-red}
  Let $(F,\Phi) \sim \psi $, and let $K$ be a positive integer such that 
  \begin{align}
K > \frac{\log(|\mcal{T}|^{n})}{e_{\bm{d}}(\psi) - \log (1 + e_{\bm{d}}(\psi))}. \label{eq:cond-K}    
  \end{align}
 Then there exist $K$ deterministic joint source-channel codes $(f_i,\varphi_i)_{ i=1}^{K}$, such that 
  \begin{align}
    \frac{1}{K}\sum_{i=1}^{K} e_{\bm{d},\theta}(f_{i},\varphi_{i}) < e_{\bm{d}}(\psi) \quad  \forall \;\theta  \in \Tscr^{n}. \label{eq:ran-code-red}
  \end{align}
\end{lemma} 

We now construct a joint source-channel code $(Q_{X|S},\varphi)$, where $Q_{X|S}$ is a stochastic encoder and $\varphi$ is a deterministic decoder. The construction is according to Theorem 12.13 in \cite{csiszar2011information}. We consider $K$ deterministic codes, $(f_{i},\varphi_{i})_{i=1}^{K}$ defined as $f_{i} : \Sscr^{k} \rightarrow \mcal{X}^{n}$ and $\varphi_{i}: \mcal{Y}^{n} \rightarrow \Sscr^{k}$, satisfying the equation \eqref{eq:ran-code-red}. A randomly chosen code from this ensemble is used to communicate the  messages from $\Sscr^{k}$. Further, we consider the deterministic code $(\wi{f},\wi{\varphi})$ defined as $\wi{f} : \{1,\hdots,K\} \rightarrow \mcal{X}^{d_n}, \; \wi{\varphi} : \mcal{Y}^{d_n} \rightarrow \{1,\hdots,K\}$ satisfying the equation \eqref{eq:kosut-kliewer-achieve}, where $d_n$ is a function of $n$. We use $(\wi{f},\wi{\varphi}) $ to communicate the index $i$ of the chosen code. Consider a random variable $\bm{i} \in \{1,\hdots,K\}$, distributed uniformly and independent of any other random variable. Given $s \in \Sscr^{k}$, the encoder chooses the input string $X \in \mcal{X}^{d_n + n}$ randomly as $ X = (\wi{f}(\bm{i}), f_{\bm{i}}(s))$. The decoder $\varphi : \mcal{Y}^{d_n+n} \rightarrow \Sscr^{k}$ decodes $y = (\wi{y},\bar{y}) \in  \mcal{Y}^{d_n + n}$ as
    \begin{align}
      \varphi(y) = \left\{ 
      \begin{array}{c l}
        s &  {\rm if} \;  (\wi{\varphi}(\wi{y}),\varphi_{i}(\bar{y})) = (i,s) \;  {\rm for \; some}\;  i\\
        0 & {\rm else}, \non
      \end{array}
\right.
    \end{align}
    where $\wi{y} \in \mcal{Y}^{d_n}$ and $\bar{y} \in \mcal{Y}^{n}$.

Using this stochastic joint source-channel code, we now bound the lower value of the zero-sum game. Note that since we encode the index of the choice of code in the input sequence of length $d_n$, the rate of communication is given as $\frac{k}{d_n+n}$.
\begin{theorem} \label{thm:minimax-val-bnd}
The minimax value of the game, $\overline{\nu}(k,n)$, is bounded above as 
  \begin{align}
&\overline{\nu}(k,n) =   \min_{\QXS,\QSY}  \max_{q } \; \mbb{P}( d(S,\wi{S}) > \bm{d})\non  \\
&\leq     \sqrt{\frac{2 \ln (3|\mcal{T}|^{d_n})}{K}} +   \min_{P_{X_{a}}} \max_{\theta_{a}} \left[ \vphantom{  \max_{\bar{x}_{a}} } \mbb{P}((X_{a},Y_{a}) \notin \mcal{A}  | \theta_{a} ) \right.  +   2K\log e \; \mbb{P}(Z(X_{a},\bar{X}_{a},Y_{a}) = 0, (X_{a},Y_{a}) \in \mcal{A} | \theta_{a}) \non\\
                    &  \left.  +  \max_{\bar{x}_{a} \in \mcal{X}^{d_n} } 2 \log 3|\mcal{T}|^{d_n} \mbb{P}(Z(X_{a},\bar{x}_{a},Y_{a}) = 0, (X_{a},Y_{a}) \in \mcal{A} | \theta_{a})  \vphantom{  \max_{\bar{x}_{a}} } \right]  \non \\
 & + \inf_{\gamma > 0}\max_{\theta_b } \left[ \mbb{E} \left[ \exp\left( -\left| i_{X;Y_{q^*}}(X_b;Y_b) -  \frac{\gamma}{P_{\wi{S}}(\Bscr_{\bm{d}}(S))} \right|^{+} \right) | \theta_b \right]  -e^{1-\gamma} \vphantom{\frac{\gamma}{P_{\wi{S}}(\Bscr_{\bm{d}}(S))}} \right],  \label{eq:minimax-val-bnd} 
  \end{align}
    where $\Xscr^{d_n} \ni \bar{X}_{a} \sim P_{X_{a}}$, $\Xscr^{d_n} \times \Yscr^{d_n} \ni (X_a,Y_a) \sim P_{X_a} P_{Y|X,\bm{\Theta}=\theta_a}$, with $\theta_a\in \Tscr^{d_n}$, $\Xscr^{n} \times \Yscr^{n} \ni (X_b,Y_{b})  \sim P_{X_b}P_{Y|X,\bm{\Theta}=\theta_b}$, with $\theta_b\in \Tscr^n$ and $P_{X_a} \in \Pscr(\Xscr^{d_n})$ and $P_{X_b} \in \Pscr(\Xscr^n)$ are any distributions.
\end{theorem}
\begin{proof}
  The maximum probability of error is bounded as follows
   \begin{align}
     &\max_{\theta \in \mcal{T}^{n + d_n}}  \;  \mbb{P}(d(S,\wi{S}) > \bm{d}  | \bm{\Theta}=\theta)  = \max_{\theta \in \mcal{T}^{n + d_n}}  \;\frac{1}{K} \sum_{i=1}^{K} \mbb{P}(d(S,\wi{S}) > \bm{d}  |  \bm{i} = i, \bm{\Theta}=\theta)   \non \\
     & \left. +   \;\frac{1}{K} \sum_{i=1}^{K} \mbb{P}(d(S,\wi{S}) > \bm{d}, \wi{\varphi}(Y_a) = i |  \bm{i} = i, \bm{\Theta}=\theta) \right) \non \\
          &\leq  \max_{\theta \in \mcal{T}^{n + d_n}} \left(  \;\frac{1}{K} \sum_{i=1}^{K} \mbb{P}(\wi{\varphi}(Y_a) \neq i  |  \bm{i} = i, \bm{\Theta}=\theta) +   \;\frac{1}{K} \sum_{i=1}^{K} \mbb{P}(d(S,\varphi_{i}(Y_b)) > \bm{d} | \bm{i} = i,  \bm{\Theta}=\theta) \right) \non \\
               &\leq  \max_{\theta_{a} \in \mcal{T}^{d_n}}   \;\frac{1}{K} \sum_{i,y_a}\mbb{I}\{\wi{\varphi}(y_a) \neq i \}P_{Y|X,\bm{\Theta}}(y_a|\wi{f}(i),\theta_a) \;  \non \\
               &+   \max_{\theta_{b} \in \mcal{T}^{n}} \frac{1}{K} \sum_{i,s,y_b}\mbb{I} \{d(s,\varphi_{i}(y_b)) > \bm{d}\} P_{S}(s) P_{Y|X,\bm{\Theta}}(y_b|f_i(s),\theta_b) . \non 
     \end{align}

Using \eqref{eq:kosut-kliewer-achieve}  from Theorem \ref{thm:kosut-kliewer-achieve} we bound the first term and using equation \eqref{eq:ran-code-red} from Lemma \ref{lem:ran-code-red} we bound the second term to get
  \begin{align}
    \max_{\theta \in \mcal{T}^{n + d_n}}  \;  \mbb{P}(d(S,\wi{S}) > \bm{d} | \bm{\Theta}=\theta)   &\leq     \sqrt{\frac{2 \ln (3|\mcal{T}|^{d_n})}{K}} +   \min_{P_{X_{a}}} \max_{\theta_{a}} \left[ \vphantom{  \max_{\bar{x}_{a}} } \mbb{P}((X_{a},Y_{a}) \notin \mcal{A}  | \theta_{a} ) \right.  \non \\
& +  \max_{\bar{x}_{a} \in \mcal{X}^{d_n} }  2 \log 3|\mcal{T}|^{d_n} \mbb{P}(Z(X_{a},\bar{x}_{a},Y_{a}) = 0, (X_{a},Y_{a}) \in \mcal{A} | \theta_{a}) \non\\
                    &  \left.  +  2K\log e \; \mbb{P}(Z(X_{a},\bar{X}_{a},Y_{a}) = 0, (X_{a},Y_{a}) \in \mcal{A} | \theta_{a}) \vphantom{  \max_{\bar{x}_{a}} } \right]  \non \\
 & + \inf_{\gamma > 0}\max_{\theta_b } \left[ \mbb{E} \left[ \exp\left( -\left| i_{X;Y_{q^*}}(X_b;Y_b) -  \frac{\gamma}{P_{\wi{S}}(\Bscr_{\bm{d}}(S))} \right|^{+} \right) | \theta_b \right]  -e^{1-\gamma} \vphantom{\frac{\gamma}{\mbb{P}(\Bscr_{\bm{d}}(S))}} \right]. \label{eq:func-of-M}
  \end{align}
Using $\max_{\theta \in \Tscr^{d_n + n}} \mbb{P}(d(S,\wi{S}) > \bm{d} | \bm{\Theta} = \theta) = \max_{q \in \Pscr(\Tscr^{d_n + n})}\mbb{P}(d(S,\wi{S}) > \bm{d})$, we get the required result. 
\end{proof}

Having derived the bounds for the lower and upper values of the game, we now proceed to derive the rate of convergence for upper and lower values of the zero-sum game. 

\section{Asymptotics and Minimax Theorems}

We now use the bounds derived in the earlier sections to compute the limits of the upper and lower values. 
Throughout this section we assume that there exists a unique capacity achieving state distribution $q_{\Theta}^{*} \in \Pi_{\Theta}$ whereby $V_{{\sf C}}^{-} = V_{{\sf C}}^{+} =: V_{\sf C}$ and $V_{\sf C} > 0$.


Let $(\mbb{X},\mbb{Y}) \sim P_{\mbb{X}} \times \sum_{\theta \in \mcal{T}} q_{\Theta}(\theta)P_{\mbb{Y}|\mbb{X},\Theta = \theta}$, with $q_{\Theta} \in \Pscr(\Tscr)$. Then for all such distributions $P_{\mbb{X}} \in \Pscr(\Xscr)$, we assume that 
\begin{align}
  i_{\mbb{X};\mbb{Y}_{q_{\Theta}^{*}}}(x;\mbb{Y}) < \infty,  \label{eq:info-dens-fin}
\end{align}
where 
 \begin{align}
       i_{\mbb{X};\mbb{Y}_{q_{\Theta}^{*}}}(x;y) &= \log \frac{(q_{\Theta}^{*}P_{\mbb{Y}|\mbb{X},\Theta})(y|x)}{(P_{\mbb{X}}q_{\Theta}^{*}P_{\mbb{Y}|\mbb{X},\Theta})(y)}, \; x \in \Xscr, y \in \Yscr \non
 \end{align}
with $q_{\Theta}^{*} \in \Pi_{\Theta}$.
Further,  we also assume that 
\begin{align}
  j_{{\sf S}}(s,\bm{d}) < \infty, \quad \forall \; s \in \Sscr \label{eq:d-tilt-fin}
\end{align}
where $j_{{\sf S}}(s,\bm{d})$ is as defined in section II-D.

\subsection{Dispersion based asymptotics of the upper bound}

In this section we compute the limit of the lower bound by using Theorem \ref{thm:minimax-val-bnd}. We first bound the error terms due to the average probability of error in equation \eqref{eq:minimax-val-bnd}. 

\subsubsection{Average error}

 Let $P^{*}_{\mbb{X}} \in \Pi_{\mbb{X}}$ and define $P^{*}_{X}(x) := \prod_{i = 1}^{d_n}P^{*}_{\mbb{X}}(x_{i})$ for $x \in \mcal{X}^{d_{n}}$.
 To use the Theorem \ref{thm:minimax-val-bnd}, we define the set $\Ascr \subseteq \mcal{X}^{d_{n}} \times \mcal{Y}^{d_{n}}$  as follows. Let $T_{\theta}^{n}(\theta') := \prod_{i=1}^{d_{n}}T_{\theta}(\theta_{i}')$ with $T_{\theta} \in \Pscr_{n}(\Tscr)$. Define
 \begin{align}
     i_{X^{*};Y_{T_{\theta}^{n}}}(x;y) &= \log \frac{(T_{\theta}^{n}P_{Y|X,\bm{\Theta}})(y|x)}{(P_{X}^{*}T_{\theta}^{n}P_{Y|X,\bm{\Theta}})(y)}, \; x \in \Xscr^{n}, y \in \Yscr^{n}. \non
 \end{align}
Then the set $\Ascr$ is defined as 
 \begin{align}
   \Ascr &= \left\{ (x,y) \in  \Xscr^{d_n} \times \Yscr^{d_n}  :  i_{X^{*};Y_{T_{\theta}^{n}}}(x;y) > \gamma    \;\;{\rm for \; some} \; T_{\theta} \in \Pscr_n(\Tscr) \vphantom{\Xscr^{d_n} \times \Yscr^{d_n} i_{X^{*};Y_{q^*}}} \right\}. \label{eq:defn-set-A} 
 \end{align}
 Further, let
 \begin{align}
{\rm V}_{0} := \sup_{q_{\Theta} \in \mcal{P}(\mcal{T})} \; {\rm Var} \left( i_{\mbb{X^{*}};\mbb{Y}_{q_{\Theta}}}(\mbb{X};\mbb{Y}) \right), \label{eq:var-bnded}
 \end{align}
 where $(\mbb{X},\mbb{Y}) \sim P_{\mbb{X}}^{*} \times \sum_{\theta \in \mcal{T}} q_{\Theta}(\theta)P_{\mbb{Y}|\mbb{X},\Theta = \theta}$. 

The following Lemma gives a necessary condition for an AVC to be non-symmetrizable. The lemma below follows from \cite[Lemma 6]{kosut2018finite}.
\begin{lemma} \label{lem:non-symm-and-joint-dist}
 Let $\mbb{X}, \mbb{X}' \sim P_{\mbb{X}}$ and $P_{\mbb{X}}(x) > 0$ for all $x \in \mcal{X}$. Let $D_{\eta}$ be the set defined as 
  \begin{align}
    D_{\eta} = \{ &Q_{\mbb{X} \mbb{X}'\Theta \mbb{Y}} \in \mcal{P}(\mcal{X} \times \mcal{X} \times \mcal{T} \times \mcal{Y}) : D(Q_{\mbb{X} \mbb{X}'\Theta \mbb{Y}} || P_{\mbb{X}} \times Q_{\mbb{X}'\Theta} \times P_{\mbb{Y}| \mbb{X},\Theta}) \leq \eta \},
\end{align}
  where $(P_{\mbb{X}} \times Q_{\mbb{X}'\Theta} \times P_{\mbb{Y}| \mbb{X},\Theta})(x,x',\theta,y) = P_{\mbb{X}}(x) Q_{\mbb{X}'\Theta}(x',\theta) P_{\mbb{Y}| \mbb{X},\Theta}(y|x,\theta)$, $Q_{\mbb{X}'\Theta}$ is the marginal of $Q_{\mbb{X} \mbb{X}'\Theta \mbb{Y}}$ and $D(Q_{\mbb{X} \mbb{X}'\Theta \mbb{Y}} || P_{\mbb{X}} \times Q_{\mbb{X}'\Theta} \times P_{\mbb{Y}| \mbb{X},\Theta})$ is the relative entropy between $Q_{\mbb{X} \mbb{X}'\Theta \mbb{Y}}$ and $ P_{\mbb{X}} \times Q_{\mbb{X}'\Theta} \times P_{\mbb{Y}| \mbb{X},\Theta}$. Let $\eta^{*}$ be defined as
  \begin{align}
    \eta^{*} = \inf \{\; \eta &: \; Q_{\mbb{X} \mbb{X}'\Theta \mbb{Y}} \in D_{\eta}, \; Q_{\mbb{X}' \mbb{X}\Theta' \mbb{Y}} \in D_{\eta} \;\; {\rm for \; some} \; Q_{\mbb{X} \mbb{X}' \Theta \Theta' \mbb{Y}} \in \mcal{P}(\mcal{X} \times \mcal{X} \times \mcal{T} \times \mcal{T}  \times \mcal{Y})\}.
  \end{align}
If the AVC is non-symmetrizable, then $\eta^{*} > 0$. 
\end{lemma}  

Using the above lemma, we now construct the function $Z : \mcal{X}^{d_{n}} \times \mcal{X}^{d_{n}} \times \mcal{Y}^{d_{n}} \rightarrow \{0,1\}$ is defined as follows. Since the AVC is non-symmetrizable, from Lemma \ref{lem:non-symm-and-joint-dist} we have $\eta^{*} > 0$. Choose $\eta$ such that $0 < \eta < \eta^{*}$. Set 
\begin{align}
  Z(x_{a},\bar{x}_{a},y_{a}) = \left\{
  \begin{array}{c l}
    1 & {\rm if} \; (x_{a},y_{a}) \in \Ascr\; \& {\rm \;either}\; (\bar{x}_{a},y_{a}) \notin \Ascr \;{\rm or\; } \exists \; \theta_{a} \in \mcal{T}^{d_{n}} \; {\rm s. \;t.}\; T_{x_{a}, \bar{x}_{a}, \theta_{a}, y_{a}} \in D_{\eta}  \\
0 & {\rm otherwise},
  \end{array} \right. \label{eq:defn-func-Z}
\end{align}
where recall that $T_{x_{a}, \bar{x}_{a}, \theta_{a}, y_{a}} $ is the joint type of $x_{a}, \bar{x}_{a}, \theta_{a}, y_{a}.$
 Thus, the function $Z$ satisfies \eqref{eq:Z-func-cond} and we have $ Z(x_{a},\bar{x}_{a},y_{a}) Z(\bar{x}_{a},x_{a},y_{a}) = 0 \;\; \forall \; x_{a} , \bar{x}_{a} \in  \mcal{X}^{d_{n}}, y_{a} \in \mcal{Y}^{d_{n}}$. Since otherwise, there exists $ (x_{a},\bar{x}_{a},y_{a}) \in \mcal{X}^{d_{n}} \times \mcal{X}^{d_{n}} \times \mcal{Y}^{d_{n}}$ such that $Z(x_{a},\bar{x}_{a},y_{a})Z(\bar{x}_{a},x_{a},y_{a}) = 1$. This implies $(x_{a},y_{a}) \in \Ascr$ and $(\bar{x}_{a},y_{a}) \in \Ascr$ and hence, there exist joint types $T_{x_{a}, \bar{x}_{a}, \theta_{a}, y_{a}} \in D_{\eta}$ and $T_{\bar{x}_{a}, x_{a}, \theta_{a}', y_{a}} \in D_{\eta}$ for some $\theta_{a}, \theta_{a}' \in \mcal{T}^{d_n}$. Thus, from the definition of $\eta^{*}$, we get $\eta^{*} \leq \eta$. However, this is a contradiction since $\eta$ is chosen to be strictly lesser than $\eta^{*}$. Using this function $Z$ we get the following upper bound. The proof is in the Appendix.


 \begin{theorem} \label{thm:avg-error-asymp-bnd}
For the above choice of $\Ascr, Z, K$ and $d_n$, we have that 
   \begin{align}
&      \sqrt{\frac{2 \ln (3|\mcal{T}|^{d_n})}{K}} +   \min_{P_{X_{a}}} \max_{\theta_{a} \in \mcal{T}^{d_n}} \left[ \vphantom{  \max_{\bar{x}_{a}} } \mbb{P}((X_{a},Y_{a}) \notin \Ascr  | \theta_{a} ) \right.  +   2K\log e \; \mbb{P}(Z(X_{a},\bar{X}_{a},Y_{a}) = 0, (X_{a},Y_{a}) \in \Ascr | \theta_{a}) \vphantom{  \max_{\bar{x}_{a}} }  \non\\
                    &  \left.  +\max_{\bar{x}_{a} \in \mcal{X}^{d_n} }  2 \log 3|\mcal{T}|^{d_n} \mbb{P}(Z(X_{a},\bar{x}_{a},Y_{a}) = 0, (X_{a},Y_{a}) \in \Ascr | \theta_{a})  \right]  \non \\
&\leq   \sqrt{\frac{2 \ln (3|\mcal{T}|^{d_n})}{K}} + \frac{ {\rm V_{0}}}{d_n\left( \delta - \frac{\log \sqrt{d_n}}{d_n}\right)^{2}} + \frac{2\log e}{\sqrt{d_n}} +  \frac{2 \log 3|\mcal{T}|^{d_n}(d_n+1)^{|\mcal{X}|^{2}|\Theta||\mcal{Y}|}}{\exp(d_n(\eta))}. \label{eq:avg-error-asymp-bnd}
   \end{align}
 \end{theorem} 

\vspace{0.5cm}
\subsubsection{Random code error}

In this section, we bound the error terms in \eqref{eq:minimax-val-bnd} corresponding to the random joint source-channel code. The proof is in the Appendix.
\begin{theorem} \label{thm:ran-code-bnd}
The random code  probability in \eqref{eq:minimax-val-bnd} is bounded as 
   \begin{align}
     & \inf_{\gamma > 0}\max_{\theta_b \in \mcal{T}^{n}} \left[ \mbb{E} \left[ \exp\left( -\left| i_{X;Y_{q^*}}(X_b;Y_b) - \log \frac{\gamma}{P_{\wi{S}}(\Bscr_{\bm{d}}(S))} \right|^{+} \right) |\theta_b \right] + e^{1-\gamma} \vphantom{ \log \frac{\gamma}{P_{\wi{S}}(\Bscr_{\bm{d}}(S))}}\right] \non \\
&\leq  {\sf Q}\left( \frac{nC - kR(\bm{d}) - \Gamma(k)}{\sqrt{n V_{{\sf C}} + k V_{{\sf S}}(\bm{d})}}\right) + \frac{B}{n+k} + \frac{K_{0} + 2}{\sqrt{k}}, \label{eq:ran-code-bnd}
   \end{align}
   where $B, K_{0} > 0$ are constants and $     \Gamma(k) = \bar{c}\log k+ c + \log \left(\frac{1}{2}\log k + 1 \right), \bar{c}, c > 0$.
 \end{theorem} 

\vspace{0.5cm}

Putting together the above two theorems, we get the following result.
\begin{theorem} \label{thm:upper-val-bnd-final}
The upper value is bounded as 
\begin{align}
 \overline{\nu}(k,n) &\leq  {\sf Q}\left( \frac{nC - kR(\bm{d}) - \Gamma(k)}{\sqrt{n V_{{\sf C}} + k V_{{\sf S}}(\bm{d})}}\right) +  \frac{B}{\sqrt{n+k}} + \frac{K_{0} + 2}{\sqrt{k}} \non \\
& + \sqrt{\frac{2 \ln (3|\mcal{T}|^{d_n})}{K}} + \frac{2\log e}{\sqrt{d_n}} + \frac{ {\rm V_{0}}}{d_n\left( \delta - \frac{\log \sqrt{d_n}}{d_n}\right)^{2}}  +  \frac{2 \log 3|\mcal{T}|^{d_n}(d_n+1)^{|\mcal{X}|^{2}|\Theta||\mcal{Y}|}}{\exp(d_n(\eta))}, \label{eq:upper-val-bnd-final}
   \end{align}
   where $ \Gamma(k)$, $K_{0}$ and $B$ are as in Theorem \ref{thm:ran-code-bnd}.
\end{theorem}
\begin{proof}
  The required bound follows by substituting the bounds in \eqref{eq:avg-error-asymp-bnd} from Theorem \ref{thm:avg-error-asymp-bnd} and \eqref{eq:ran-code-bnd} from Theorem \ref{thm:ran-code-bnd}  in \eqref{eq:minimax-val-bnd}. 
\end{proof}
We now choose a sequence of $(k,n)$ for which the upper and lower values of the game tend to zero. Take
\begin{align}
K = c_{0}n, \; d_n = \ceil[\bigg]{\frac{\log K}{C - \delta}}, \label{eq:defn-K-and-dn}
\end{align}
 where $ c_{0} \in \mbb{N}$ is such that $K$ satisfies the equation \eqref{eq:cond-K}, $ \delta > 0$ and $\ceil{\bullet}$ is the ceiling function.

The following corollary gives the dispersion bound for the rate of communication.
\begin{corollary}
\textit{[Dispersion based achievability]}  Let $\epsilon > 0$. Consider a sequence of $(k,n)$ satisfying
  \begin{align}
    & nC - kR(\bm{d}) - \Gamma(k)  \geq \sqrt{n V_{{\sf C}} + k V_{{\sf S}}(\bm{d})} \;{\sf Q}^{-1} \left( \epsilon \right).
  \end{align}
Then, we have that  $\underline{\nu}(k,n) \leq \overline{\nu}(k,n) \leq   \epsilon$.
\end{corollary}
\begin{proof}
  Take $(k,n)$ such that   $nC - kR(\bm{d})  \geq \sqrt{n V_{{\sf C}} + k V_{{\sf S}}(\bm{d})} \;{\sf Q}^{-1} \left( \epsilon -  \delta(k,n)  \right) +  \Gamma(k)$, where
  \begin{align}
    \delta(k,n) &= \frac{B}{\sqrt{n+k}} + \sqrt{\frac{2 \ln (3|\mcal{T}|^{d_n})}{K}}  + \frac{2\log e}{\sqrt{d_n}}   + \frac{ {\rm V_{0}}}{d_n\left( \delta - \frac{\log \sqrt{d_n}}{d_n}\right)^{2}} +  \frac{2 \log 3|\mcal{T}|^{d_n}(d_n+1)^{|\mcal{X}|^{2}|\Theta||\mcal{Y}|}}{\exp(d_n(\eta))} \non
  \end{align}
  and $K, d_n$ are as in \eqref{eq:defn-K-and-dn}.
  Substituting  in \eqref{eq:upper-val-bnd-final}, we get that
  \begin{align}
& {\sf Q}\left( \frac{nC - kR(\bm{d}) - \Gamma(k)}{\sqrt{n V_{{\sf C}} + k V_{{\sf S}}(\bm{d})}}\right) +  \frac{B}{\sqrt{n+k}} + \sqrt{\frac{2 \ln (3|\mcal{T}|^{d_n})}{K}}  + \frac{2\log e}{\sqrt{d_n}} + \frac{ {\rm V_{0}}}{d_n\left( \delta - \frac{\log \sqrt{d_n}}{d_n}\right)^{2}}  \non \\
      &+  \frac{2 \log 3|\mcal{T}|^{d_n}(d_n+1)^{|\mcal{X}|^{2}|\Theta||\mcal{Y}|}}{\exp(d_n(\eta))} \leq \epsilon.
  \end{align}
  Since $\delta(k,n) \rightarrow 0$ as $k,n \rightarrow \infty$, the result follows.
\end{proof}

\subsection{Dispersion based asymptotics of the lower bound}

In this section, we compute the limit of the lower bound of the probability of error. The proof is in the Appendix.

 \begin{theorem} \label{thm:conv-bnd-asymp}
The lower value $\underline{\nu}(k,n)$ is bounded as
    \begin{align}
\underline{\nu}(k,n) &\geq  \max_{q,P_{\overline{Y}_{q}}, {\sf U}} \; \sup_{\gamma > 0}  \left[  \sum_{s}P_{S}(s) \right.  \min_{x} \left[ \vphantom{ \sum_{u=1}^{{\sf U}}} \mbb{P}\left( j_{S}(s,\bm{d}) - i_{X;\overline{Y}_{q}|U}(x;Y|U) \leq \gamma \right)\right. \non \\
     &  + \exp ( j_{S}(s,\bm{d})-\gamma) \sum_{u=1}^{{\sf U}} \sum_{y}P_{U|X}(u|x) P_{\overline{Y}_{q}|U}(y|u) \left. \left. \mbb{I}\left\{  j_{S}(s,\bm{d}) - i_{X;\overline{Y}_{q}|U}(x;y|u) > \gamma \right\}  \vphantom{ \sum_{u=1}^{{\sf U}}} \right]- \frac{{\sf U}}{\exp(\gamma)} \right] \non \\
      &\geq  {\sf Q}\left( \frac{nC - kR(\bm{d}) + \gamma(n)}{\sqrt{nV_{{\sf C}} + kV_{{\sf S}}(\bm{d}) - K_{3}}}\right) - \frac{K_{1}}{k} - \frac{K_{2}}{\sqrt{n}} - \frac{B'}{\sqrt{n+k}} - \frac{1}{\sqrt{n+1}}, \label{eq:conv-bnd-asymp}
    \end{align}
    where $K_{1}, K_{2}, K_{3}, B' > 0$ are constants and $ \gamma(n) = K_{4}\log (n+1), K_{4} > 0$.
\end{theorem} 

We also have the following corollary that gives the second order dispersion bounds on the rate.
\begin{corollary}
\textit{[Dispersion based converse]}  Let $\epsilon > 0$. Consider a sequence of $(k,n)$ satisfying
  \begin{align}
    & nC - kR(\bm{d})  \leq \sqrt{n V_{{\sf C}} + k V_{{\sf S}}(\bm{d})} \;{\sf Q}^{-1} \left( \epsilon \right) -  \gamma(n).
  \end{align}
Then, we have that $  \overline{\nu}(k,n) \geq  \underline{\nu}(k,n) \geq  \epsilon$.
 \end{corollary}
 
\begin{proof}
  Take $(k,n)$ such that $ nC - kR(\bm{d}) + \gamma(n)  \leq \sqrt{n V_{{\sf C}} + k V_{{\sf S}}(\bm{d})} \;{\sf Q}^{-1} \left( \epsilon +  \delta(k,n)  \right)$,  where 
  \begin{align}
    \delta(k,n) &=  \frac{K_{1}}{k} + \frac{K_{2}}{\sqrt{n}} + \frac{1}{\sqrt{n+1}} + \frac{B'}{\sqrt{n+k}}. \non
  \end{align}
  Substituting in \eqref{eq:conv-bnd-asymp}, we get that
  \begin{align}
& {\sf Q}\left( \frac{nC - kR(\bm{d}) + \gamma(n)}{\sqrt{nV_{{\sf C}} + kV_{{\sf S}}(\bm{d}) - K_{3}}}\right) - \frac{K_{1}}{k} - \frac{K_{2}}{\sqrt{n}} - \frac{1}{\sqrt{n+1}} - \frac{B'}{\sqrt{n+k}} \leq \epsilon. 
  \end{align}
  Since $\delta(k,n) \rightarrow 0$ as $k,n \rightarrow \infty$, the result follows.
\end{proof}

\subsection{Proof of the minimax theorems}

We are now prepared to derive the asymptotic results stated in Section \ref{sec:form}. We consider sequences of $(k,n)$ and study the limiting behaviour of the finite blocklength bounds along these sequences. 

We first prove Theorem \ref{thm:minimax-tends-0}.

\begin{proof}[\textit{of Theorem \ref{thm:minimax-tends-0}}]
Consider the bound from Theorem \ref{thm:ran-code-bnd}.  Take a sequence $(k,n) \uparrow \infty$ such that $\lim \frac{k}{n} < \frac{C}{R(\bm{d})}$. Then, we have that
  \begin{align}
    & \lim_{k,n \rightarrow \infty} {\sf Q}\left( \sqrt{n}\left(\frac{C - \frac{k}{n}R(\bm{d}) + \frac{\Gamma(k)}{n} }{\sqrt{ V_{{\sf C}} + \frac{k}{n} V_{{\sf S}}(\bm{d})}}\right)\right) = 0. \non
  \end{align}
Further, as $k,n \rightarrow \infty$ and for $K$ and $d_n$ chosen as \eqref{eq:defn-K-and-dn}, we have
\begin{align}
& \frac{B}{n+k} + \frac{K_{0}+2}{\sqrt{k}} + \sqrt{\frac{2 \ln (3|\mcal{T}|^{d_n})}{K}} + \frac{2\log e}{\sqrt{d_n}}   + \frac{ {\rm V_{0}}}{d_n\left( \delta - \frac{\log \sqrt{d_n}}{d_n}\right)^{2}} +  \frac{2 \log 3|\mcal{T}|^{d_n}(d_n+1)^{|\mcal{X}|^{2}|\Theta||\mcal{Y}|}}{\exp(d_n(\eta))} \rightarrow 0. \non
\end{align}
Thus, $ \overline{\nu}(k,n) \rightarrow 0$ and hence $\underline{\nu}(k,n) \rightarrow 0$.
\end{proof}

We now come to the proof of Theorem \ref{thm:minimax-tends-1}.

\begin{proof}[\textit{of Theorem \ref{thm:minimax-tends-1}}]
Consider the bound in Theorem \ref{thm:conv-bnd-asymp}.  Take a sequence $(k,n) \uparrow \infty$ such that $\lim \frac{k}{n} > \frac{C}{R(\bm{d})}$. Then, we have that
  \begin{align}
      \lim_{k,n \rightarrow \infty} {\sf Q}\left( \sqrt{n} \left( \frac{C - \frac{k}{n}R(\bm{d}) + \frac{\gamma(n)}{n}}{\sqrt{V_{{\sf C}} + \frac{k}{n}V_{{\sf S}}(\bm{d}) - \frac{K_{3}}{n}}} \right)\right) = 1. \non
  \end{align}
Further, as $k,n \rightarrow \infty$
\begin{align}
   - \frac{K_{1}}{k} - \frac{K_{2}}{\sqrt{n}} - \frac{1}{\sqrt{n+1}} - \frac{B'}{\sqrt{n+k}} \rightarrow 0. \non 
\end{align}
Thus, $\underline{\nu}(k,n) \rightarrow 1 $ and hence $\overline{\nu}(k,n) \rightarrow 1 $.
\end{proof}
Note that proving the coincidence of the upper and lower values as in Theorem~\ref{thm:minimax-tends-0} and Theorem~\ref{thm:minimax-tends-1} can be shown without dispersion bounds (see \eg, the arguments in~\cite{vora2019minimax}). Consequently, these claims hold even when the capacity achieving state distribution is not unique.

To prove Theorem \ref{thm:minimax-para-L}, we consider a refined definition of the rate with $O(\frac{1}{\sqrt{n}})$ backoff from $\frac{C}{R(\bm{d})}$. Recall from  \eqref{eq:k-n-seq-L} that the sequence is given as 
\begin{align}
  \frac{k}{n} = \frac{C}{R(\bm{d})} + \frac{\rho}{\sqrt{n}}, \label{eq:k-n-seq-L-proof}
\end{align}
where $\rho \in \mbb{R}$ is fixed. The non-asymptotic formulae we have derived allow us to evaluate the upper and lower value of the game along this refined definition of the rate. Along this sequence, the upper and lower value now both tend to the same value, which depends on $\rho$. This value is in general neither $0$ or $1$. This coincidence in finer asymptotics hinges on there being a unique capacity achieving distributions for the jammer. We discuss this in the following section.

\begin{proof}[\textit{of Theorem \ref{thm:minimax-para-L}}]
Substituting \eqref{eq:k-n-seq-L-proof} in the lower bound given by \eqref{eq:conv-bnd-asymp}, we get the following Gaussian term in the inequality
    \begin{align}
&    {\sf Q}\left( \frac{nC - kR(\bm{d}) + \gamma(n)}{\sqrt{nV_{{\sf C}} + kV_{{\sf S}}(\bm{d}) - K_{3}}}\right) 
    = {\sf Q}\left( \frac{- \rho R(\bm{d}) + \frac{\gamma(n)}{\sqrt{n}}}{\sqrt{V_{{\sf C}} + (\frac{C}{R(\bm{d})} + \frac{\rho}{\sqrt{n}})V_{{\sf S}}(\bm{d}) - \frac{K_{3}}{n}}}\right).
    \end{align}
  Taking limit $k,n \rightarrow \infty$, we get
    \begin{align}
   & \lim_{k,n \rightarrow \infty } \underline{\nu}(k,n)  \geq  {\sf Q}\left( \frac{- \rho R(\bm{d}) }{\sqrt{V_{{\sf C}} + \frac{C}{R(\bm{d})}V_{{\sf S}}(\bm{d})}}\right),  
    \end{align}
    since the other terms in \eqref{eq:conv-bnd-asymp} tend to zero asymptotically.
        
Similarly, substituting \eqref{eq:k-n-seq-L-proof} in \eqref{eq:upper-val-bnd-final}, we get the following Gaussian term in the inequality.
\begin{align}
  &{\sf Q}\left( \frac{nC - kR(\bm{d}) - \Gamma(k)}{\sqrt{ nV_{{\sf C}} + k V_{{\sf S}}(\bm{d})}}\right) 
= {\sf Q}\left(   \frac{ -\rho R(\bm{d}) - \frac{\Gamma(k)}{\sqrt{n}} }{\sqrt{ V_{{\sf C}} + (\frac{C}{R(\bm{d})} + \frac{\rho}{\sqrt{n}})V_{{\sf S}}(\bm{d})}}\right). 
\end{align}
Taking the limit as $k,n \rightarrow \infty$, we get 
 \begin{align}
   & \lim_{k,n \rightarrow \infty } \underline{\nu}(k,n)  \geq  {\sf Q}\left( \frac{- \rho R(\bm{d}) }{\sqrt{V_{{\sf C}} + \frac{C}{R(\bm{d})}V_{{\sf S}}(\bm{d})}}\right),  
    \end{align}
    since the other terms in \eqref{eq:upper-val-bnd-final} vanish asymptotically. This completes the proof.
\end{proof}

\subsection{Non-uniqueness of channel dispersion}

In a general AVC setting, there could be multiple pairs of distributions $(P_{\mbb{X}},q_{\Theta})$ that achieve the capacity in \eqref{eq:avc-cap} and consequently the sets $\Pi_{\mbb{X}}$ and $\Pi_{\Theta}$ could be non-singleton sets. When the capacity achieving input and the state  distributions are non-unique, the channel dispersions $ V_{{\sf C}}^{+} $ and $ V_{{\sf C}}^{-}$ defined in Section II-D need not be equal. The dispersions are equal if either one of the capacity achieving input or state distributions is unique.

 In this paper, we computed the dispersion bounds for stochastic joint source-channel coding over an AVC under the assumption that the capacity achieving state distribution is unique.  Recently, authors in \cite{kosut2018finite} and \cite{kosut2017dispersion} computed the second order dispersion bounds for the rate of communication over an AVC for deterministic and random codes respectively. They showed that in the case of non-unique capacity achieving distributions, the achievability and converse bounds for the rate are not equal. Thus, it is natural to suppose that in the case of non-unique capacity achieving distributions, the `finer' minimax theorem may not hold and thereby 
 the achievability and converse bounds for the rate may not match in the joint source-channel setting. Validating this would require improved joint source-channel coding strategies which is a point of future work.


\section{Conclusion}

We formulated the adversarial communication problem as zero-sum game between the encoder-decoder team and the jammer where the encoder-decoder attempt to minimize the probability of error while the jammer tries to maximize it. The problem is non-convex in the space of strategies of the encoder-decoder team and hence a minimax theorem need not hold. However, we showed that an approximate minimax theorem holds for the game. We derived finite blocklength upper and lower bounds for the minimax and maximin values of the game and showed that asymptotic minimax theorems hold as the blocklength tends to infinity. In particular, for rates below $\frac{C}{R(\bm{d})}$, the upper and lower values tend to zero and for rates above $\frac{C}{R(\bm{d})}$, the values tend to unity. For rates tending exactly to $\frac{C}{R(\bm{d})}$ with a $O(\frac{1}{\sqrt{n}})$ backoff, the values tend to the same constant under a technical assumption on the uniqueness of capacity achieving distributions.

\begin{appendices}
\section{}
We begin with the following central limit theorem due to Berry and Esseen (see ~\cite{polyanskiy2010channel}).
 \begin{theorem}[Berry-Esseen CLT] \label{thm:berry-esseen}
 Fix $n \in \mbb{N}$. Let $W_{i}$ be independent random variables. Then, for $t \in \mbb{R}$, we have
   \begin{align}
     \left| \mbb{P}\left( \frac{1}{n}\sum_{i=1}^{n} W_{i} >  D_{n} + t\sqrt{\frac{V_{n}}{n}} \right) - {\sf Q}(t)\right| \leq \frac{B_{n}}{\sqrt{n}},
   \end{align}
   where ${\sf Q}$ is the complementary Gaussian function, and
   \begin{align}
     D_{n} &= \frac{1}{n}\sum_{i=1}^{n}\mbb{E}[W_{i}], \;\; V_{n} =  \frac{1}{n}\sum_{i=1}^{n}{\sf Var}[W_{i}], \;\;
     A_{n} =  \frac{1}{n}\sum_{i=1}^{n}\mbb{E}[|W_{i} - \mbb{E}[W_{i}|^{3}]], \;\; B_{n} = \frac{c_{0}A_{n}}{V_{n}^{3/2}}, c_{0} >0. 
   \end{align}
 \end{theorem}

\begin{proof}[\textit{of Theorem \ref{thm:avg-error-asymp-bnd}}]

We weaken the bound in \eqref{eq:minimax-val-bnd} by taking $P_{X_{a}}(x_{a}) = \prod_{i=1}^{d_n} P^{*}_{\mbb{X}}(x_{i}), P_{\mbb{X}}^* \in \Pi_{\mbb{X}}$. Let $\Ascr, Z, K$ and $d_n$ be as defined in \eqref{eq:defn-set-A}, \eqref{eq:defn-func-Z} and \eqref{eq:defn-K-and-dn} respectively.
Let $  U_{i} := i_{\mbb{X}^{*};\mbb{Y}_{T_{\theta}}}(X_{ai};Y_{ai})$, where $(X_{ai}, Y_{ai}) \sim P^{*}_{\mbb{X}} \times \sum_{\theta \in \mcal{T}}T_{\theta}(\theta)P_{\mbb{Y}|\mbb{X},\Theta = \theta} \; \forall \;i$. Since the channel is memoryless, we have $i_{X^{*};Y_{q^*}}(X_a;Y_a) = \sum_{i=1}^{d_{n}}U_{i}$. 
     Thus, taking $\gamma=\log(\sqrt{d_n}K)$ we can write the following, $\mbb{P}\left(i_{X^{*};Y_{q^*}}(X_a;Y_a) \leq \gamma \right) =  
\mbb{P} \left( \sum_{i=1}^{d_n}U_{i} \leq \log (\sqrt{d_n}K ) \right)$
           From Lemma 12.10 in  \cite{csiszar2011information}, we have $ C \leq  \mbb{E} [i_{\mbb{X};\mbb{Y}_{T_{\theta}^{n}}}(X_{ai};Y_{ai})] = E[U_{i}]$ for all $T_{\theta} \in \Pscr_{n}(\Tscr)$ and hence $ d_nC \leq   \sum_{i = 1}^{d_n} \mbb{E}[U_{i}]$. Also, substituting for $\log K$ from equation \eqref{eq:defn-K-and-dn}, we get
     \begin{align}
   \mbb{P}\left( \sum_{i=1}^{d_n}U_{i} \leq \log \sqrt{d_n} + \log K  \right) 
   &\leq    \mbb{P}\left( \sum_{i=1}^{d_n} U_{i} \leq \log \frac{\sqrt{d_n}}{\exp  d_n\delta}  +  \sum_{i = 1}^{d_n} \mbb{E}[U_{i}]   \right) \non\\
        &\leq          \mbb{P}\left( \left| \sum_{i=1}^{d_n} \left(U_{i} - \mbb{E}[U_{i}]\right) \right| \geq  \log \frac{\exp  d_n\delta}{\sqrt{d_n}}  \right),  
        \end{align}
where the last equation follows from triangle inequality. Using Chebyshev's inequality and taking supremum over $\theta_{a}$, we get
        \begin{align}
\sup_{\theta_{a} \in \mcal{T}^{d_n}}   \mbb{P}\left( \left| \sum_{i=1}^{d_n} (U_{i} - \mbb{E}[U_{i}]) \right|  \geq  d_n\delta - \log \sqrt{d_n} \vphantom{\sum_{i=1}^{d_n} (i_{\mbb{X};\mbb{Y}_{q^{*}}}(X_{ai};Y_{ai})} \right) &\leq   \sup_{\theta_{a} \in \mcal{T}^{d_n}} \frac{d_n}{(d_n \delta - \log \sqrt{d_n})^{2}} {\rm Var} \left( i_{\mbb{X}^{*};\mbb{Y}_{T_{\theta_{a}}}}(\mbb{X}_{a};\mbb{Y}_{a}) \right) \non \\
   &\leq  \frac{V_{0}}{d_n\left( \delta - \frac{\log \sqrt{d_n}}{d_n}\right)^{2}}, \label{eq:v0-proof} 
\end{align}
where \eqref{eq:v0-proof} follows from equation \eqref{eq:var-bnded}.

The bounds on the other two terms given as $\mbb{P}\left(Z(X_{a},\bar{X}_{a},Y_{a} \right) = 0, (X_{a},Y_{a}) \in \Ascr | \theta_{a}) \leq \frac{1}{\sqrt{d_n}K} $ and $\mbb{P}(Z(X_{a},\bar{x}_{a},Y_{a}) = 0, (X_{a},Y_{a}) \in \Ascr | \bar{X}_{a} = \bar{x}_{a}, \theta_{a}) \leq (d_n+1)^{|\mcal{X}|^{2}|\Theta||\mcal{Y}|} \exp(-d_n\eta)$ follows from the proof of Theorem 4 in \cite{kosut2018finite}. 
Using \eqref{eq:v0-proof} and above bounds, we  get the required bound.  
\end{proof}

\begin{proof}[\textit{of Theorem \ref{thm:ran-code-bnd}}]
To derive the bound, we construct a random joint source-channel code $(F,\Phi)$ or equivalently a distribution $\psi$ on the set of codes $\{ f,\varphi \}$. From Theorem \ref{thm:joint-random-code-achieve}, it suffices to choose the distributions $P_{X}$ and $P_{\wi{S}}$. Let $P^{*}_{\mbb{X}}$ be a distribution from the set $\Pi_{\mbb{X}}$ that achieves the $V_{{\sf C}}^{+}$ in \eqref{eq:chan-disp}. Define $P_{X}(x) := \prod_{i = 1}^{n}P^{*}_{\mbb{X}}(x_{i})$ for $x \in \mcal{X}^{n}$. Further, take $P_{\wi{S}}(\wi{s})  = \prod_{i=1}^{k}P_{\wi{\mbb{S}}}^{*}(\wi{s}_{i}), \wi{s} \in \Sscr^{k}$,  where $P_{\wi{\mbb{S}}}^{*}$ achieves the optimal in \eqref{eq:rate-dist-func}. Clearly, with these choice of distributions, we have that 
\begin{align}
i_{X^{*};Y_{q^*}}(X;Y) &= \sum_{i=1}^{n} i_{\mbb{X}^{*};\mbb{Y}_{q_{\Theta}^{*}}}(X_{i};Y_{i}), \;\;  j_{S}(S,\bm{d}) = \sum_{j=1}^{k}j_{{\sf S}}(S_{j},\bm{d}).  \non
\end{align}

Recall the error terms corresponding to the random code from Theorem \ref{thm:minimax-val-bnd}.
   Writing the maximization over $\theta_b \in \Tscr^n$ as maximization over $q \in \Pscr(\Tscr^n)$, we can write the bound as
   \begin{align}
 &   \max_{q \in \Pscr(\mcal{T}^{n})} \left[ \mbb{E} \left[ \exp\left( -\left| i_{X^{*};Y_{q^*}}(X_b;Y_b) - \log \frac{\gamma}{P_{\wi{S}}(\Bscr_{\bm{d}}(S))} \right|^{+} \right)  \right]  + e^{1 - \gamma} \vphantom{\frac{\gamma}{P_{\wi{S}}(\Bscr_{\bm{d}}(S))}} \right], \non 
   \end{align}
   where $q(\theta) = \prod_{i=1}^{n}q_i(\theta_i)$ with $q_i \in \Pscr(\Tscr), \theta \in \Tscr^n$.
Let $ h(X_b,Y_b,S) :=  \sum_{j=1}^{n}i_{\mbb{X}^{*};\mbb{Y}_{q_{\Theta}^{*}}}(X_{bj};Y_{bj}) - \log \frac{\gamma}{P_{\wi{S}}(\Bscr_{\bm{d}}(S))}$
Further, define the set $\Dscr$ as
\begin{align}
  \mcal{D} :=& \left\{ s \in \Sscr^{k} : \log \frac{1}{P_{\wi{S}}(\Bscr_{\bm{d}}(s))}  \leq \sum_{i=1}^{k}j_{{\sf S}}(s_{i},\bm{d}) + \left( \bar{c} - \frac{1}{2} \right) \log k  + c \right\},
\end{align}
where $\bar{c}, c$ are constants defined in \cite{kostina2013lossy}. Define the random variable $W_{l}$ as 
\begin{align}
 W_{l} = W_l(n,k):= \begin{cases} 
i_{\mbb{X}^{*};\mbb{Y}_{q_{\Theta}^{*}}}(X_{bl};Y_{bl}) & \eef l \leq n \\
 -  j_{{\sf S}}(S_{l-n},\bm{d}) & \eef n<l \leq n+k 
\end{cases} \label{eq:defn-Wl}
\end{align}
The expectation $\mbb{E} \left[ \exp\left( -\left| h(X_b,Y_b,S) \right|^{+} \right)  \right]$ can be written as 
\begin{align}
  &   \mbb{E} \left[ \vphantom{\frac{\gamma}{P_{Z}(\Bscr_{\bm{d}}(S))}} \exp\left( -\left| h(X_b,Y_b,S) \right|^{+} \right)  \mbb{I}\left\{ S \in \mcal{D} \right\}  \right] +  \mbb{E} \left[ \vphantom{\frac{\gamma}{P_{Z}(\Bscr_{\bm{d}}(S))}} \exp\left( -\left| h(X_b,Y_b,S) \right|^{+} \right)\mbb{I}\left\{ S \notin \mcal{D} \right\}  \right] \non \\
  &\leq  \mbb{E} \left[ \exp\left( -\left|  \sum_{l=1}^{n+k}  W_{l} -  \log (k^{\left( \bar{c} - \frac{1}{2}\right)}\exp(c)\gamma)  \right|^{+} \right)  \right] + \frac{K_{0}}{\sqrt{k}}, \label{eq:prob-bnd-set-D}
\end{align}
where the first term in  \eqref{eq:prob-bnd-set-D}  follows by using the definition of the set $\Dscr$ and the second term follows by using $\mbb{E} \left[ \exp\left( -\left| h(X_b,Y_b,S) \right|^{+} \right)\mbb{I}\left\{ S \notin \mcal{D} \right\}  \right] \leq \mbb{P}(S \notin \Dscr)$  and applying Lemma 5 from \cite{kostina2013lossy}.

We define the following moments to be used in bounding the first term in \eqref{eq:prob-bnd-set-D} using Berry-Esseen CLT.
\begin{equation}
\begin{aligned}
D_{n+k}(q) &= \frac{1}{n+k}\sum_{l=1}^{n+k}\mbb{E}[W_{l} ],  \;\;
V_{n+k}(q) =  \frac{1}{n+k}\sum_{l=1}^{n+k}{\sf Var}[W_{l} ],  \\
A_{n+k}(q) &=  \frac{1}{n+k}\sum_{l=1}^{n+k}\mbb{E}[|W_{l} - \mbb{E}[W_{l}|^{3} ],  \;\;
B_{n+k}(q) = \frac{c_{0}A_{n+k}(q)}{V_{n+k}^{3/2}(q)}, c_{0} >0. \label{eq:moments-Wl}
\end{aligned}  
\end{equation}
Note that the moments are computed with respect to the distribution $ P_{X}\times \sum_{\theta} q(\theta)P_{Y|X,\bm{\Theta} = \theta} \times P_{S}$.  Next, we define the following set
        \begin{align}
  \mcal{H} &= \left\{ \vphantom{\sqrt{\frac{V_{n+k}}{n+k}}} (x,y,s) \in \Xscr^{n} \times \Yscr^{n} \times \Sscr^{k} : \frac{1}{n+k}\sum_{l=1}^{n+k}W_{l}   > \left(D_{n+k}(q) - t_{k,n}\sqrt{\frac{V_{n+k}(q)}{n+k}}\right)\right\}, 
   \end{align}
where $t_{k,n} > 0$ will be chosen later.
For the sake of brevity, we define the term in the $\exp$ as follows
   \begin{align}
     g(X_b,Y_b,S) &= \sum_{l=1}^{n+k}  W_{l} -  \log (k^{\left( \bar{c} - \frac{1}{2}\right)}\exp(c)\gamma), \label{eq:defn-func-g}\\
\Gamma_{n+k}(q) &= (n+k)\left(D_{n+k}(q) - t_{k,n}\sqrt{\frac{V_{n+k}(q)}{n+k}}\right) -  \log (k^{\left( \bar{c} - \frac{1}{2} \right)} \exp(c)\log \gamma). \label{eq:defn-big-gamma}
   \end{align}

   Thus, we can write \eqref{eq:prob-bnd-set-D} as $ \mbb{E} \left[ \exp\left( -\left| g(X_b,Y_b,S)  \right|^{+} \right) \right]$ which can be written as 
   \begin{align}
     &  \mbb{E} \left[ \exp\left( -\left| g(X_b,Y_b,S) \right|^{+} \right) \mbb{I}\left\{ (X_b,Y_b,S) \in \mcal{H} \right\}  \right] +  \mbb{E} \left[ \exp\left( -\left| g(X_b,Y_b,S)  \right|^{+} \right) \mbb{I}\left\{ (X_b,Y_b,S) \notin \mcal{H} \right\}  \right] \non \\
     &\leq  \mbb{E} \left[ \exp\left( -\left| \Gamma_{n+k}(q)  \right|^{+} \right) \mbb{I}\left\{ (X_b,Y_b,S) \in \mcal{H} \right\} \right] +  \mbb{E} \left[  \mbb{I}\left\{ (X_b,Y_b,S) \notin \mcal{H} \right\}   \right],      \label{eq:exp-term-bnd-1}
   \end{align}
   where the first term in \eqref{eq:exp-term-bnd-1} follows by using the definition of the set $\mcal{H}$ and the second term follows by using $\exp(-|\bullet|^{+}) \leq 1$. Using the above results and \eqref{eq:prob-bnd-set-D}, we get the following bound.
   \begin{align}
     &   \max_{q \in \Pscr(\mcal{T}^{n})} \left[ \mbb{E} \left[ \exp\left( -\left| i_{X^{*};Y_{q^{*}}}(X_b;Y_b) - \log \frac{\gamma}{P_{\wi{S}}(\Bscr_{\bm{d}}(S))} \right|^{+} \right)  \right]  + e^{1 - \gamma} \vphantom{\frac{\gamma}{P_{\wi{S}}(\Bscr_{\bm{d}}(S))}} \right] \non \\
     &\leq  \max_{q \in \Pscr(\mcal{T}^{n})}\mbb{E} \left[ \exp\left( -\left| \Gamma_{n+k}(q)  \right|^{+} \right) \mbb{I}\left\{ (X_b,Y_b,S) \in \mcal{H} \right\}  \right] + \max_{q \in \Pscr(\mcal{T}^{n})} \mbb{P} \left(   (X_b,Y_b,S) \notin \mcal{H}  \right) +  e^{1 - \gamma} + \frac{K_{0}}{\sqrt{k}}.      \label{eq:bnd-before-clt}
   \end{align}
   To upper bound the first term in \eqref{eq:bnd-before-clt}, we compute the maximum of $ \exp\left( -\left| \Gamma_{n+k}(q)  \right|^{+} \right)$ over all distributions. Since $\exp(-|\bullet|^+)$ is a decreasing function with respect to $\Gamma_{n+k}(q)$, we get
   \begin{align}
     & \max_{q \in \Pscr(\mcal{T}^{n})}\exp\left( -\left| \Gamma_{n+k}(q)  \right|^{+} \right) \leq \exp\left( -\left| \min_{q \in \Pscr(\mcal{T}^{n})}\Gamma_{n+k}(q)  \right|^{+} \right).
   \end{align}
To compute the minimum, we consider the following 
   \begin{align}
  D_{n+k}(q) &= \frac{1}{n+k}\sum_{i=1}^{n}\mbb{E}[i_{\mbb{X}^{*};\mbb{Y}_{q_{\Theta}^{*}}}(\mbb{X}_{bi};\mbb{Y}_{bi})] - \frac{k}{n+k} R(\bm{d}), \non \\
     V_{n+k}(q) &=  \frac{1}{n+k} \sum_{i=1}^{n}{\sf Var}(i_{\mbb{X}^{*};\mbb{Y}_{q_{\Theta}^{*}}}(\mbb{X}_{bi};\mbb{Y}_{bi})) + \frac{k}{n+k} V_{{\sf S}}(\bm{d}), \non 
   \end{align}
where the moments  are with respect to $P_{\mbb{X}}^*\times \sum_{\theta \in \Tscr} q_{i}(\theta)P_{\mbb{Y}|\mbb{X},\Theta = \theta}$.
Thus, the minimum is given as 
   \begin{align}
     & \min_{q \in \Pscr(\Tscr^n)}(n+k) \left( D_{n+k}(q) - t_{k,n}\sqrt{\frac{V_{n+k}(q)}{n+k}} \right) =  nC - k R(\bm{d})  - (n+k)\max_{q_i \in \Pi_{\Theta}} t_{k,n}\sqrt{\frac{V_{n+k}(q)}{n+k}} + O(1), \label{eq:max-of-sum}
   \end{align}
   where the maximum in the second term in \eqref{eq:max-of-sum} is restricted to $\Pi_{\Theta}$ by using Lemma 63 and Lemma 64 in \cite{polyanskiy2010channel}. 
   Since $P_{\mbb{X}}^{*} \in \Pi_{\mbb{X}}$ and $q_{\Theta}^{*} \in \Pi_{\Theta}, |\Pi_{\Theta}| = 1$, we have $\max_{q_i \in \Pi_{\Theta}} V_{n+k}(q) = nV_{{\sf C}}.$ Thus, we have that
   \begin{align}
     & \min_{q \in \Pscr(\Tscr^n)}(n+k) \left( D_{n+k}(q) - t_{k,n} \sqrt{\frac{V_{n+k}(q)}{n+k}} \right)  = nC - kR(\bm{d}) -  t_{k,n} \sqrt{n V_{{\sf C}} + k V_{{\sf S}}(\bm{d})}  + O(1),
   \end{align}
   We choose $t_{k,n}$ as  $t_{k,n} = (nC - kR(\bm{d})  - \bar{c}\log k - \log \gamma - c)/ \sqrt{n V_{{\sf C}} + k V_{{\sf S}}(\bm{d})}$. Thus, we get  \newline $  \min_{q \in \Pscr(\Tscr^{n})}\Gamma_{n+k}(q)  \geq \frac{1}{2}  \log k + O(1).$ Substituting in \eqref{eq:bnd-before-clt}, we get the following upper bound
\begin{align}
  & \max_{q \in \Pscr(\Tscr^{n})}\frac{1}{\sqrt{k}}\mbb{P} \left( (X_b,Y_b,S) \in \mcal{H} \right) + \max_{q \in \Pscr(\Tscr^{n})} \mbb{P} \left(   (X_b,Y_b,S) \notin \mcal{H}  \right) +  e^{1 - \gamma} +  \frac{K_{0}}{\sqrt{k}}\non  \\
       &\leq   {\sf Q}(t_{k,n})  + \max_{q \in \Pscr(\Tscr^{n})}\frac{B_{n+k}(q)}{\sqrt{n+k}} + \frac{K_{0} + 2}{\sqrt{k}}, \label{eq:berry-esseen-bnd}
\end{align}
where \eqref{eq:berry-esseen-bnd} follows by taking $\gamma = \frac{1}{2}\log_{e}k + 1$  and  bounding $\mbb{P} \left(  (X_b,Y_b,S) \notin \mcal{H}  \right)$ using Berry-Esseen CLT. 

Recall that $B_{n+k}(q)$ is given as $ B_{n+k}(q) = \frac{c_{0}A_{n+k}(q)}{V_{n+k}^{3/2}(q)}$. Assuming $\min_{l \in \{1,\hdots,n+k\}}{\sf Var}[W_{l}] \neq 0$, we can  bound $B_{n+k}(q)$ as from the definition of $A_{n+k}(q)$ and $V_{n+k}(q)$ as
\begin{align}
   B_{n+k}(q) \leq \frac{c_{0} \max_{l} \mbb{E}[|W_{l} - \mbb{E}[W_{l}]|^{3}] }{(\min_{l}{\sf Var}[W_{l}])^{3/2}}. \label{eq:Bnk-bnd}
\end{align}
From \eqref{eq:info-dens-fin} and \eqref{eq:d-tilt-fin}, we have that the RHS is  finite for all $q$ and independent of $k,n$. Thus, we have that  $\max_{q \in \Pscr(\Tscr^{n})} B_{n+k}(q)$ is a finite constant. Substituting $t_{k,n}$ from above, taking $ \max_{q \in \Pscr(\Tscr^{n})} B_{n+k}(q) \leq B$ with $B > 0$ and using \eqref{eq:berry-esseen-bnd}, we get the required bound.
\end{proof}

\begin{proof}[\textit{of Theorem \ref{thm:conv-bnd-asymp}}]
From \eqref{eq:dp-bound}, it suffices to construct $q, P_{\bar{Y}_{q}}$, the random variable $U$ and $\gamma$ to get a lower bound on $\underline{\nu}(k,n)$. 
Take $ q(\theta) = q^{*}(\theta) = \prod_{i=1}^{n} q^{*}_{\Theta}(\theta_{i})$ where $q^{*}_{\Theta} \in \Pi_{\Theta}$. Let ${\sf U}$ be the number of types in $\Pscr_{n}(\Xscr)$ and $\Uscr = \{1, \hdots, {\sf U}\}$ be indices corresponding to each of the types. Thus, for a given sequence $x \in \Xscr^{n}$, $U$ maps it to its type which is denoted by some index $u \in \Uscr$. Further, let $P_{\overline{Y}_{q}|U}$ be defined as
\begin{align}
  P_{\overline{Y}_{q}|U}(y|u) &= (P_{X}q^{*}P_{Y|X,\bm{\Theta}})(y) = \sum_{x,\theta}  P_{X}(x)q^{*}(\theta)P_{Y|X,\bm{\Theta}}(y|x,\theta),
\end{align}
where $P_{X}(x) = \prod_{i=1}^{n}T_{x}(x_{i}),  x \in \Xscr^n$ with $T_{x} \in \Pscr_{n}(\Xscr)$ being the type corresponding to the index $u$. Thus, we have that $i_{X;\overline{Y}_{q}|U}(x;Y|u) = \sum_{i=1}^{n} i_{\mbb{X};\mbb{Y}_{q^{*}}}(x_{i};Y_{i}), \;\;
  j_{S}(s,\bm{d}) = \sum_{j=1}^{k}j_{{\sf S}}(s_{j},\bm{d}), $
where 
\begin{align}
 i_{\mbb{X};\mbb{Y}_{q^{*}}}(x';y)  := \log \frac{(q_{\Theta}^{*}P_{\mbb{Y}|\mbb{X},\Theta})(y|x')}{(T_{x}q_{\Theta}^{*}P_{\mbb{Y}|\mbb{X},\Theta})(y)}, \; x' \in \Xscr, y \in \Yscr.
\end{align}
Since $q$ is taken as an i.i.d. distribution, effectively, we have a channel where \newline $Y \sim \prod_{i=1}^{n} \sum_{\theta_{i} \in \Tscr}q_{\Theta}^{*}(\theta_{i})P_{\mathbb{Y}|\mathbb{X} = x_{i},\Theta = \theta_{i}}$ when the input is $x = (x_1, \hdots, x_n)$. Thus, the LHS of \eqref{eq:conv-bnd-asymp} is a converse of the standard DMC without a jammer with the channel given as the above averaged channel.  Following the line of arguments given in Appendix C of \cite{kostina2013lossy}, we get the following inequality.
\begin{align}
   & \max_{q,P_{\overline{Y}_{q}}, {\sf U}} \; \sup_{\gamma > 0}  \left[  \sum_{s}P_{S}(s) \right.  \min_{x} \left[ \vphantom{ \sum_{u=1}^{{\sf U}}} \mbb{P}\left( j_{S}(s,\bm{d}) - i_{X;\overline{Y}_{q}|U}(x;Y|U) \leq \gamma \right)\right. \non \\ 
                    &  + \exp ( j_{S}(s,\bm{d})-\gamma) \sum_{u=1}^{{\sf U}} \sum_{y}P_{U|X}(u|x) P_{\overline{Y}_{q}|U}(y|u) \left. \left. \mbb{I}\left\{  j_{S}(s,\bm{d}) - i_{X;\overline{Y}_{q}|U}(x;y|u) > \gamma \right\}  \vphantom{ \sum_{u=1}^{{\sf U}}} \right]- \frac{{\sf U}}{\exp(\gamma)} \right] \non \\
   &\geq \mbb{P}\left( \sum_{i=1}^{n}i_{\mbb{X};\mbb{Y}_{q_{\Theta}^{*}}}(x_{i}^{*};Y_{i}) - \sum_{j=1}^{k} j_{{\sf S}}(S_{j},\bm{d})) \leq -\gamma \right)  - \frac{K_{1}}{k} - \frac{K_{2}}{\sqrt{n}} - (n+1)^{|\Xscr|-1} \exp(-\gamma), \label{eq:conv-bnd-before-clt}
\end{align}
where the second term in the LHS is bounded below by zero, $K_{1}$ and $K_{2}$ are some constants and  $x^{*} = (x_{1}^{*},\hdots,x_{n}^{*})$ is a sequence such that its type $T_{x^{*}}$ minimizes
\begin{align}
\min_{T_{x} \in \Pscr_{n}(\Xscr)}  \left|T_{x} - P_{\mbb{X}}^{*} \right| \label{eq:x*-defn}
\end{align}
where $P_{\mbb{X}}^{*} \in \Pi_{\mbb{X}}$. Let
\begin{align}
W_{l} = W_{l}(n,k) := \left\{
  \begin{array}{c l}
 i_{\mbb{X};\mbb{Y}_{q_{\Theta}^{*}}}(x_{l}^{*};Y_{l})     &  \eef l \leq n\\
j_{{\sf S}}(S_{n-l},\bm{d}) &   \eef   n < l \leq n+k. 
  \end{array} \right. \non
\end{align}
Define the following moments of the random variable $W_{l}$. 
\begin{align}
  D_{n+k} &= \frac{1}{n+k}\sum_{l=1}^{n+k}\mbb{E}[W_{l}], \;\;
  V_{n+k} =  \frac{1}{n+k}\sum_{l=1}^{n+k}{\sf Var}[W_{l}],  \non \\
 A_{n+k} &=  \frac{1}{n+k}\sum_{l=1}^{n+k}\mbb{E}[|W_{l} - \mbb{E}[W_{l}|^{3}]],  \;\;
  B_{n+k}' = \frac{c_{0}A_{n+k}}{V_{n+k}^{3/2}}, c_{0} >0. \non
\end{align}  
From Berry-Esseen CLT, we have
\begin{align}
  \mbb{P}\left( \sum_{l=1}^{n+k}W_{l} \leq -\gamma \right) \geq {\sf Q}\left( \frac{D_{n+k} + \frac{\gamma}{n+k}}{\sqrt{\frac{V_{n+k}}{n+k}}}\right) - \frac{B_{n+k}'}{\sqrt{n+k}}. \label{eq:CLT-bnd-conv}
\end{align}
We also have the following inequalities from the Appendix C of \cite{kostina2013lossy}. 
\begin{align}
    D_{n+k}   &\leq \frac{n}{n+k}C - \frac{k}{n+k}R(\bm{d}),   \label{eq:conv-cap-rate-ineq}\\
  V_{n+k}   &\geq \frac{n}{n+k}V_{{\sf C}} + \frac{k}{n+k}V_{{\sf S}}(\bm{d}) - \frac{K_{3}}{n+k}, \label{eq:conv-disp-ineq}
\end{align}
where $K_3 > 0$  is some constant. Further, from \eqref{eq:info-dens-fin} and \eqref{eq:d-tilt-fin},  we can show that $A_{n+k}$ is bounded and hence $B'_{n+k}$ is bounded by a constant $B' > 0$.
Using \eqref{eq:conv-cap-rate-ineq} and \eqref{eq:conv-disp-ineq}, we get 
\begin{align}
  &{\sf Q}\left( \frac{D_{n+k} + \frac{\gamma}{n+k}}{\sqrt{\frac{V_{n+k}}{n+k}}}\right) \geq  {\sf Q}\left( \frac{nC - kR(\bm{d}) + \gamma}{\sqrt{nV_{{\sf C}} + kV_{{\sf S}}(\bm{d}) - K_{3}}}\right). \non 
\end{align}
Substituting the above in \eqref{eq:CLT-bnd-conv}, taking $\gamma = \left(|\Xscr| - \frac{1}{2}\right)\log (n+1)$ and using \eqref{eq:conv-bnd-before-clt}, we get the required bound.  
\end{proof}

\end{appendices}

\bibliographystyle{IEEEtran}
\bibliography{ref.bib}
\end{document}

%% file: macros-ieee.tex
\usepackage{amsmath, amssymb, bbm, xspace}
\usepackage{epsfig}
\usepackage{longtable}
\usepackage{color}
\usepackage{mathrsfs}
\usepackage{comment}
\usepackage{ifthen}
\newboolean{showcomments}
\setboolean{showcomments}{true}
\newcommand{\aak}[1]{  \ifthenelse{\boolean{showcomments}}
{ \textcolor{blue}{(AAK says:  #1)}} {}  }

\usepackage{courier}

\newtheorem{theorem}{Theorem}[section]

\newtheorem{lemma}[theorem]{Lemma}

\newtheorem{corollary}[theorem]{Corollary}

\def\bkE{{\rm I\kern-.17em E}}
\def\bk1{{\rm 1\kern-.17em l}}
\def\bkD{{\rm I\kern-.17em D}}
\def\bkR{{\rm I\kern-.17em R}}
\def\bkP{{\rm I\kern-.17em P}}

\def\bkZ{{\bf{Z}}}

\def\bkE{{\rm I\kern-.17em E}}
\def\bk1{{\rm 1\kern-.17em l}}
\def\bkD{{\rm I\kern-.17em D}}
\def\bkR{{\rm I\kern-.17em R}}
\def\bkP{{\rm I\kern-.17em P}}

\makeatletter
\newcommand{\pushright}[1]{\ifmeasuring@#1\else\omit\hfill$\displaystyle#1$\fi\ignorespaces}
\newcommand{\pushleft}[1]{\ifmeasuring@#1\else\omit$\displaystyle#1$\hfill\fi\ignorespaces}
\makeatother

\def\bkZ{{\bf{Z}}}
\def\b12{(\beta_1,\beta_2)}

\newenvironment{example}{{\noindent \bf Example}}{\hfill $\square$\hspace{-4.5pt}\vspace{6pt}}
\newcounter{example}
\renewcommand{\theexample}{\thesection.\arabic{example}}

\newcounter{remark}
\renewcommand{\theremark}{\thesection.\arabic{remark}}

\def\Bscr{\mathscr{B}}
\def\Xscr{\mathcal{X}}
\def\Yscr{\mathcal{Y}}

\newlength{\noteWidth}
\long\def\notes#1{\ifinner
{\tiny #1}
\else
\marginpar{\parbox[t]{\noteWidth}{\raggedright\tiny #1}}
\fi\typeout{#1}}

 \def\notes#1{\typeout{read notes: #1}} 

\newcommand{\I}[1]{\mathbb{I}_{\{#1\}}}

\newcommand{\ie}{i.e.\@\xspace} 
\newcommand{\eg}{e.g.\@\xspace} 
\newcommand{\etal}{et al.\@\xspace} 

\newcommand{\Real}{\ensuremath{\mathbb{R}}}

\newcommand{\minimize}[1]{\displaystyle\minim_{#1}}
\newcommand{\minim}{\mathop{\hbox{\rm min}}}
\newcommand{\maximize}[1]{\displaystyle\maxim_{#1}}
\newcommand{\maxim}{\mathop{\hbox{\rm max}}}

\def\OPT{{\rm OPT}}

\def\Nbb{{\mathbb{N}}}

\def\ast{\buildrel{t}\over\rightarrow}

\def\exp{\mathop{\hbox{\rm exp}}}

\def\spose#1{\hbox to 0pt{#1\hss}}
\def\sub#1{^{\null}_{#1}}
\def\text #1{\hbox{\quad#1\quad}}

\def\nthinsp{\mskip -2   mu}

\def\superstar{^{\raise 0.5pt\hbox{$\nthinsp *$}}}
\def\SUPERSTAR{^{\raise 0.5pt\hbox{$*$}}}

\def\lamstarT {\lambda^{\raise 0.5pt\hbox{$\nthinsp *$}T}}

\def\Ascr{{\cal A}}
\def\Bscr{{\cal B}}

\def\Dscr{{\cal D}}

\def\Tscr{{\cal T}}

\def\Pscr{{\cal P}}

\def\Sscr{{\cal S}}

\def\Uscr{{\cal U}}

\def\Wscr{{\cal W}}

\def\Xscr{{\cal X}}
\def\Yscr{{\cal Y}}

\def\shat{\widehat s}

\def\eef{\;\textrm{if}\;}

\def\non{\nonumber}

\let\forallnew\forall
\renewcommand{\forall}{\forallnew\ }
\let\forall\forallnew

		\def\bkE{{\rm I\kern-.17em E}}
		\def\bk1{{\rm 1\kern-.17em l}}
		\def\bkD{{\rm I\kern-.17em D}}
		\def\bkR{{\rm I\kern-.17em R}}
		\def\bkP{{\rm I\kern-.17em P}}
		\def\bkY{{\bf \kern-.17em Y}}
		\def\bkZ{{\bf \kern-.17em Z}}
		\def\bkC{{\bf  \kern-.17em C}}

{\begin{list}{}%
         {\setlength{\leftmargin}{#1}}%
         \item[]%
}
{\end{list}}

		\def\bsp{\begin{split}}
		\def\beq{\begin{eqnarray}}
		\def\bal{\begin{align*}}
		\def\bc{\begin{center}}
		\def\be{\begin{enumerate}}
		\def\bi{\begin{itemize}}
		\def\bs{\begin{small}}
		\def\bS{\begin{slide}}
		\def\ec{\end{center}}
		\def\ee{\end{enumerate}}
		\def\ei{\end{itemize}}
		\def\es{\end{small}}
		\def\eS{\end{slide}}
		\def\eeq{\end{eqnarray}}
		\def\eal{\end{align*}}
		\def\esp{\end{split}}
		\def\qed{ \vrule height7.5pt width7.5pt depth0pt}  

\def\maxproblemsmallwb#1#2#3#4{
		 \begin{tabular*}{0.47\textwidth}
			{@{}l@{\extracolsep{\fill}}l@{\extracolsep{-4pt}}l@{\extracolsep{\fill}}c@{}}
				#1 &  & $\maximize{#2}$  $#3$ & $ $ \\[4pt]
					  $\sub \ $  &   & $#4$ &  $ $
			\end{tabular*}
			}

\def\problemsmallawb#1#2#3#4{
		 \begin{tabular*}{0.47\textwidth}
			{@{}l@{\extracolsep{\fill}}l@{\extracolsep{-4pt}}l@{\extracolsep{\fill}}c@{}}
				#1 &  & $\minimize{#2}$  $#3$ & $ $ \\[4pt]
					  $\sub \ $  &   & $#4$ &  $ $
			\end{tabular*}
                      }
\def\problemsmallabb#1#2#3#4{\fbox{
		 \begin{tabular*}{0.55\textwidth}
			{@{}l@{\extracolsep{\fill}}l@{\extracolsep{-4pt}}l@{\extracolsep{\fill}}c@{}}
				#1 &  & \quad $= $ $\minimize{#2}$  $#3$&   \\[4pt]
                        $\sub \ $  &   & $#4$ &  $ $
			\end{tabular*}}
			}
\def\maxproblemsmallbb#1#2#3#4{\fbox{
		 \begin{tabular*}{0.55\textwidth}
			{@{}l@{\extracolsep{\fill}}l@{\extracolsep{-4pt}}l@{\extracolsep{\fill}}c@{}}
				#1 &  & \quad $ = $ $\maximize{#2}$  $#3$ & $ $ \\[4pt]
					  $\sub \ $  &   & $#4$ &  $ $
			\end{tabular*}
			}}

	\def\cp2problem#1#2#3#4{\fbox
		 {\begin{tabular*}{0.9\textwidth}
			{@{}l@{\extracolsep{\fill}}l@{\extracolsep{6pt}}l@{\extracolsep{\fill}}c@{}}
				#1 & & $#4 $ 
			\end{tabular*}}}

		\def\bkE{{\rm I\kern-.17em E}}
		\def\bk1{{\rm 1\kern-.17em l}}
		\def\bkD{{\rm I\kern-.17em D}}
		\def\bkR{{\rm I\kern-.17em R}}
		\def\bkP{{\rm I\kern-.17em P}}
		
		\def\bkZ{{\bf{Z}}}

\newcommand {\beeq}[1]{\begin{equation}\label{#1}}
\newcommand {\eeeq}{\end{equation}}
\newcommand {\bea}{\begin{eqnarray}}
\newcommand {\eea}{\end{eqnarray}}

\def\texitem#1{\par\smallskip\noindent\hangindent 25pt
               \hbox to 25pt {\hss #1 ~}\ignorespaces}

\def\bsp{\begin{split}}
		\def\beq{\begin{eqnarray}}
		\def\bal{\begin{align*}}
		\def\bc{\begin{center}}
		\def\be{\begin{enumerate}}
		\def\bi{\begin{itemize}}
		\def\bs{\begin{small}}
		\def\bS{\begin{slide}}
		\def\ec{\end{center}}
		\def\ee{\end{enumerate}}
		\def\ei{\end{itemize}}
		\def\es{\end{small}}
		\def\eS{\end{slide}}
		\def\eeq{\end{eqnarray}}
		\def\eal{\end{align*}}
		\def\esp{\end{split}}
		\def\qed{ \vrule height7.5pt width7.5pt depth0pt}  

\usepackage{amsmath, amssymb, xspace}
\usepackage{epsfig}
\usepackage{longtable}
\usepackage{color}
\usepackage{mathrsfs}
\usepackage{subfig}
\newenvironment{proof}[1][]{{\noindent \emph {Proof} #1: }}{\hfill \qed \vspace{3pt}\\ }

\def\sub{\hbox{\rm s.t}}

